\documentclass[journal]{IEEEtran}
\usepackage{times,epsfig,graphicx,wrapfig,bbm,color,amssymb}

\newcommand {\cA}{{\mathcal{A}}}
\newcommand {\cD}{{\mathcal{D}}}

\newcommand {\cX}{{\mathcal{X}}}
\newcommand {\cY}{{\mathcal{Y}}}
\newcommand {\tC}{\overline{SC}}

\newcommand {\bC} {{\bf C}}

\newcommand {\bN} {{\bf N}}
\newcommand {\bX} {{\bf X}}

\newcommand {\bff} {{\bf f}}
\newcommand {\bu} {{\bf u}}
\newcommand {\bx} {{\bf x}}
\newcommand {\by} {{\bf y}}
\newcommand {\bm} {{\bf m}}

\newcommand{\avg}{{\rm avg}}

\newcommand {\Z} {{\mathbbm Z}}
\newcommand {\N} {{\rm I\kern-1.5pt N}}
\newcommand {\R} {{\rm I\kern-2.5pt R}}

\newtheorem{theorem}{Theorem}
\newtheorem{assm}{Assumption}

\newcommand{\beqa}{\begin{eqnarray}}
\newcommand{\eeqa}{\end{eqnarray}}
\newcommand{\beqan}{\begin{eqnarray*}}
\newcommand{\eeqan}{\end{eqnarray*}}
\newcommand{\beq}{\begin{equation}}
\newcommand{\eeq}{\end{equation}}

\newcommand{\bfl}{\begin{flushleft}}
\newcommand{\efl}{\end{flushleft}}
\newcommand{\bgnv}{\begin{verbatim}}
\newcommand{\ednv}{\end{verbatim}}

\newcommand{\myb}{\hspace{-0.1in}}

\newcommand{\myeq}{& \hspace{-0.1in} = & \hspace{-0.1in}}

\newcommand{\lb}{\nonumber \\}
\newcommand{\myarr}{\begin{array}{lll}}

\newcommand{\mygeq}{& \myb \geq & \myb}

\newcommand{\myleq}{& \myb \leq & \myb}
\newcommand{\myl}{& \myb < & \myb}
\newcommand{\bitem}{\begin{itemize}}
\newcommand{\eitem}{\end{itemize}}
\newcommand{\benum}{\begin{enumerate}}
\newcommand{\eenum}{\end{enumerate}}

\newcommand{\E}[1]{{\mathbbm E}\left[ #1 \right]}
\newcommand{\bP}[1]{{\mathbbm P}\left[ #1 \right]}
\newcommand{\myhb}{\hspace{-0.3in}}

\newcommand{\myf}{\hspace{0.1in}}

\def\QED{~\rule[-1pt]{5pt}{5pt}\par\medskip}
\newenvironment{proof}{{\bf Proof: \ }}{ \hfill \QED}

\newcommand{\ER}{Erd$\ddot{\rm o}$s-R$\acute{\rm e}$nyi }

\begin{document}

\title{Interdependent Security with Strategic Agents and Cascades of 
Infection}

\author{Richard J. La\thanks{This work was 
supported in part by the National Science Foundation
under Grant CCF 08-30675 and a grant from National Institute
of Standards and Technology.} 
\thanks{Author is with the Department of Electrical \& 
Computer Engineering (ECE) and the Institute for Systems 
Research (ISR) at the University of Maryland, College Park.
E-mail: hyongla@umd.edu}
}

\maketitle

\begin{abstract}
We investigate cascades in networks consisting of strategic agents
with interdependent security. We assume that the {\em strategic} agents have 
choices between i) investing in protecting themselves, ii) purchasing 
insurance to transfer (some) risks, and iii) taking no actions. 
Using a population game model, we study how various system parameters, 
such as node degrees, infection propagation rate, and the probability 
with which infected nodes transmit infection to neighbors, affect nodes'
choices at Nash equilibria and the resultant price of anarchy/stability. 
In addition, we examine how the probability that a single infected node can
spread the infection to a significant portion of the entire network, called 
{\em cascade probability}, behaves with respect to system parameters. 
In particular, we demonstrate that, at least for some parameter regimes, the 
cascade probability increases with the average degree of nodes. 
\end{abstract}

\begin{IEEEkeywords}
Cascade, contagion, interdependent security, population game, price of anarchy. 
\end{IEEEkeywords}

\IEEEpeerreviewmaketitle

\section{Introduction}	\label{sec:Introduction}

Recently, the topic of {\em interdependent security} (IDS) 
\cite{HealKun2004} has gained 
much attention from research communities. IDS arises naturally in many
areas including cybersecurity, airline security, and smart power grid,
just to name a few. Ensuring adequate security of such critical 
infrastructure and systems has emerged as one of most important
engineering and societal challenges today. 

There are several key difficulties in tackling IDS in large networks. 
First, as the name suggests, the security 
of individual entities is dependent on those of others. Second,
these entities are often strategic and are interested only in their own
objectives with little or no regards for the well being of the others.
Third, any attempt to capture and study detailed interactions among 
a large number of  
(strategic) entities suffers from the {\em curse of dimensionality}. 

Although there are no standard metrics on which experts agree for 
measuring or quantifying system-level security, 
one popular approach researchers take to measure the security of a 
network is to see
how easily an infection can spread through a network. In particular,
researchers often study the probability with which an infection will
propagate to a significant or nonnegligible fraction of the
network, starting with a {\em single} infected node in the network, 
which we call {\em cascade probability}. 

We study the cascade probability in a network composed of
{\em strategic} 
agents or nodes representing, for instance, 
organizations (e.g., companies) or network domains. The edges in 
the network are not necessarily physical edges. Instead, they could 
be {\em logical, operational} or {\em relational} edges (e.g., business 
transactions or information sharing). The degree of a 
node is defined to be the number of neighbors or incident edges it 
has in the network.\footnote{We 
assume that the network is modeled as an {\em undirected} 
graph in the paper. If the network was modeled as a {\em directed}
graph instead, the degree distribution of players we are interested 
in would be that of in-degrees.}  

In our setting, there are malicious entities, called {\em attackers}, 
which launch attacks against the nodes in the network, for example, in 
hopes of infecting the machines or gaining unauthorized access to 
information of victims. Moreover, when the attacks are successful, 
their victims also unknowingly launch {\em indirect} attacks on 
their neighbors. For this reason, when a node is vulnerable to attacks, 
it also heightens the risk of its neighbors as well, thereby 
introducing a {\em negative network externality} and influencing the 
choices of its neighbors. Network externality is also known as 
{\em network effect} \cite{ShapiroVarian}.

Faced with the possibility of being attacked either directly by 
malicious attackers or indirectly by their neighbors, nodes may
find that it is in their own interests to invest in protecting 
themselves against possible attacks, e.g., firewalls, network 
intrusion detection tools, incoming traffic monitoring, etc. 
Moreover, they may also consider purchasing 
insurance to mitigate their financial losses in 
case they fall victim to successful attacks. To capture these
choices available to nodes, 
in our model each node can select from three admissible 
actions -- Protection ($P$), Insurance ($I$) and No Action ($N$). 
When a node picks $N$, it assumes all of the risk from damages
or losses brought on by successful attacks. 

In practice, a node may be able to both invest in protecting itself 
and purchase insurance at the same time. However, because insurance 
merely {\em transfers} risk from the insured to the insurer, a purchase 
of insurance by a node that also invests in protection does not
affect the preferences of other nodes. Therefore, 
not modeling the possibility of simultaneous investment in protection and 
insurance by a node does not change other nodes' decisions to protect 
themselves. Moreover, both overall social costs, i.e., the sum of losses 
due to attacks and investments in protection, and cascade probability 
depend only on which nodes elect to invest in protection. Therefore, 
leaving out the
choice of simultaneous protection and insurance does not 
alter our main findings 
on the price of anarchy/stability (POA/POS) \cite{KoutPapa} and cascade 
probability, which are explained shortly. 

As mentioned earlier, a major hurdle to studying IDS in a large network 
consisting of many nodes is that it is difficult, if not impossible, 
to model the details of interactions among all nodes. To skirt this
difficulty, we employ a {\em population game} model 
\cite{Sandholm}. A population
game is often used to model the interactions between many players, 
possibly from different populations. 
While the population game model is clearly a simplification of a complicated 
reality, we believe that our findings based on this {\em scalable}
model offer helpful insights into more realistic scenarios.  

For our study, we adopt a well known solution concept, {\em Nash 
equilibrium} (NE) of the population game, as an approximation to 
nodes' behavior in practice. 
Our goal is to investigate how the {\em network effects} present
in IDS shape the POA/POS
and cascade probability as different
system parameters (e.g., node degree distribution and 
infection propagation rate) are varied.

The POA (resp. POS) is defined to be the {\em largest} (resp. 
{\em smallest}) ratio between the social cost at an NE and the smallest
achievable social cost. The POS can be viewed as the minimum price one 
needs to pay for {\em stability} among the players so that no player
would have an incentive to deviate from its strategy unilaterally. 
Both POS and POA have recently gained much attention as a means to 
measure the inefficiency of NEs for different games 
(e.g., \cite{KoutPapa, AGT, Rough2005}).

Our main findings and contributions can be summarized as follows:
\benum

\item There exists a threshold on degree of populations so that only the
populations with degree greater than or equal to the threshold invest in 
protection. This degree threshold decreases with an increasing 
propagation rate of infection and the probability of indirect attacks 
on neighbors.

\item In general, there may not be a unique NE of a population game. 
However, the size of each population investing in protection is 
identical at all NEs. Consequently, the overall social cost 
and cascade probability are the same for all NEs, and the POA and the 
POS are identical.

\item We provide an upper bound on the POA/POS, which is a function of 
the {\em average} degree of populations and increases superlinearly
with the average degree in many cases. Moreover, it is tight in the 
sense that we can find scenarios for which the POA is equal to the bound. 

\item In many cases, the population size investing in protection 
tends to climb with the average degree, the infection propagation 
rate, and the probability of indirect attack on neighbors. 
Somewhat surprisingly, the cascade probability also increases
at the same time as the average degree or indirect
attack probability rises. 

We suspect that this observation is a consequence of the following: 
As more of the population invests in protection, it 
produces {\em higher positive network externalities} 
on other unprotected nodes. These greater
positive externalities in turn cause {\em free riding} by 
some nodes with larger degrees which would have chosen to protect 
when the parameters were smaller. These vulnerable nodes 
with larger degrees then provide better venues for 
an infection to spread, escalating the cascade probability as a 
result.  

\eenum

We point out that our analysis of cascades is carried out under a 
simplifying assumption that local neighborhoods of nodes are 
tree-like. While this assumption is reasonable for {\em sparse} 
networks, it may not hold in some of real-world networks that have 
been shown to exhibit much stronger {\em local clustering} than 
many of random graph models~\cite{Seshadhri2012, WattsStrogatz1998}. 
For such networks with higher clustering among neighbors and 
cycles in local neighborhoods, our findings may not be directly 
applicable. 

In addition, we note that actual security investments in 
practice will likely vary considerably from one realized network
to another, even when the node degrees are identical. Moreover, 
in some cases, the exact network topology may be unknown and hard 
to obtain. For these reasons, it is difficult and, perhaps, 
uninformative to study the effects of node degree distribution 
on the basis of a limited number of random networks. Instead, we
aim to capture the {\em mean} behavior of nodes without fixing the 
network topology. It is our hope that even this simple model will 
help us understand the {\em qualitative} nature of aggregate behavior
of nodes and shed some light on how the underlying structure of 
interdependency in security shapes their security decisions and 
resulting network-level security in more realistic settings.

To the best of our knowledge, our work presented here (along with 
\cite{La_TON2013}, in which we explore {\em local} network
security seen by individual nodes and a structural relation 
between an NE and a social optimum) 
is the first study to investigate
the effects of network properties and other system parameters on 
interdependent security in networks of {\em strategic} entities. 
Although our study is based on a
population game model that does not capture {\em microscopic}
strategic 
interactions among individual nodes, we believe that it approximates
the {\em macroscopic} behavior of the nodes and our findings 
shed some light on how the underlying network topology and other 
system parameters may influence the choices of nodes in practice and 
shape the resulting network security.

The rest of the paper is organized as follows. We summarize some
of most closely related studies in Section~\ref{sec:Related}. 
Section \ref{sec:Model} outlines the population game model we adopt 
for our analysis and presents the questions of interest to us. 
Section \ref{sec:MainResults} discusses our main analytical results 
on the properties of NEs and the POA/POS, which are complemented
by numerical results in Section \ref{sec:Numerical}. We conclude in 
Section \ref{sec:Conclusion}.

\section{Related literature}		\label{sec:Related}

Due to a large volume of literature related to security and cascades of 
infection, an attempt to summarize the existing studies will be an
unproductive exercise. 
Instead, we only select several key studies that are most relevant 
to our study and discuss them briefly. Furthermore, for a summary of
related literature on IDS, we refer an interested reader to 
\cite{La_TON2013, Laszka}
and references therein. 
Here, we focus on the literature related to cascades and contagion.

First, Watts in his seminal paper \cite{Watts2002} 
studied the following question: Consider a network with $n$ nodes whose 
degree distribution is given by ${\bf p} = (p_k; \ k \in 
\Z_+)$, where $\Z_+ := \{0, 1, 2, \ldots\}$. 
Suppose that we randomly choose a single node
and infect it. Given this, what is the probability that a large
number of nodes will be infected, starting with the single 
infected node, i.e., there is a (global) cascade of infection? 
Obviously, the answer to this question depends on how the infection 
spreads. In Watts' model, each node $i$ has a random threshold $\Theta_i 
\in [0, 1]$, and it becomes infected once the fraction of its neighbors 
that are infected exceeds $\Theta_i$. 

In his analysis, rather than deriving a global cascade condition for 
finite networks, he considers an infinite network in which each 
node has degree $k$ with probability $p_k$, independently of others. 
Using a generating function approach, he then 
studies the condition under which the largest cluster of vulnerable 
nodes percolates, which he calls the {\em cascade condition}. Here, 
a node is vulnerable if its 
threshold is smaller than the inverse of its degree. 

A somewhat surprising finding in his study is that 
as the average or mean degree of nodes increases, 
the network goes through two critical (phase) transitions: Initially, when the 
average degree is very small, the network is stable in that the cascade 
probability is (near) zero. As the average degree climbs, 
after the first transition the network experiences cascades 
with nonnegligible probability. However, as the average degree rises 
further, at some point, the network becomes stable again and cascades
do not occur frequently, i.e., the cascade probability becomes very 
small once again. 

Gleeson and Cahalane \cite{Gleeson2007} extended the work of Watts.
In their model, they assumed that a certain fraction of total population 
is infected at the beginning, and showed that the existence of 
cascades exhibits high sensitivity to the size of initially infected 
population. 

Coupechoux and Lelarge examined the influence of clustering 
in social networks on diffusions and contagions in random networks
\cite{CoupLelarge1, CoupLelarge2}. In particular,
they proposed a new random graph model where some nodes are 
replaced by cliques of size equal to the degrees of the nodes. Their key findings
include the observation that, in the symmetric threshold model, the effects of 
clustering on contagion threshold depend on the mean node degree; for small mean 
degrees, clustering impedes contagion, whereas for large mean degrees, contagion\
is facilitated by clustering.

An observation similar to Watts' finding has been reported in different fields, 
including financial markets where banks and financial institutions 
(FIs) are interconnected through their overlapping investment 
portfolios and other (credit) exposures
\cite{Beale2011, Caccioli2011, Caccioli2012, GaiKapadia}. 
In a simple model \cite{Caccioli2011}, two 
FIs are connected if they share a common asset in 
their investment portfolios, and the average degree of FIs
depends on the number of available assets and how
diverse their portfolios are, i.e., how many assets each FI owns.
An interesting finding is that when the number of overlapping assets of 
FIs is small, the market is stable in that it 
can tolerate a failure of a few FIs without affecting other FIs 
significantly. As they begin to diversify their 
portfolios and spread
their investments across a larger set of assets, the market becomes
unstable in that a failure of even one or two FIs triggers a domino
effect, causing many other FIs to collapse shortly after. 
However, when they 
diversify their investment portfolios even further and include a very 
large set of assets, the market becomes stable again.

Watts' model has also been extended to scenarios where nodes are connected
by more than one type of network, e.g., social network vs. 
professional network \cite{Brummitt2012, Yagan2012}. 
For example, Ya$\breve{\rm g}$an and 
Gligor \cite{Yagan2012} investigated scenarios where
nodes are connected via two
or more networks with varying edge-level influence. In their model, 
each node switches from ``good'' to ``infected'' when
$\sum c_i \cdot m_i / \sum c_i \cdot k_i$ exceeds some
threshold, where $k_i$ and $m_i$ are the total number of neighbors
and infected neighbors, respectively, 
of the node in the $i$th network, and $c_i$ 
reflects the relative influence of the edges in the $i$th network.  
Their main finding related to the impact of average degree is similar 
in nature to that of Watts~\cite{Watts2002}. 

In another related study, Beale et al. \cite{Beale2011} studied the 
behavior of strategic banks interested in minimizing their own 
probabilities of failure. They showed that banks can lower own probability
of failure by diversifying their risks and spreading across assets. But, 
if banks follow similar diversification strategies, it can cause a (nearly) 
simultaneous collapse of multiple banks, thereby potentially compromising 
the stability of the whole financial market. This finding points to a tension 
between the stability of individual banks and that of the financial system. 
Although the authors did not attempt to quantify the loss of stability, 
this degradation in system stability is closely related to well 
known inefficiency of NEs \cite{Dubey, AGT, Pound}.  

We point out an important difference between the findings in the studies 
by Watts and others \cite{Caccioli2011, Caccioli2012, Watts2002} and ours: 
In our model, the nodes are strategic and can actively protect themselves 
when it is in their own interests to do so. 
In such scenarios, as the 
average degree increases, in many cases the network becomes more vulnerable 
in that the cascade probability rises despite that more nodes protect 
themselves (Section~\ref{sec:Numerical}). This somewhat counterintuitive 
observation is a sharp departure 
from the findings of \cite{Watts2002, Yagan2012}. 

This discrepancy is mainly caused by the following. In the model 
studied by Watts and others, the thresholds of nodes are given by 
independent and identically distributed (i.i.d.) random variables 
(rvs), and their distribution does {\em not} depend on the average
degree or node degrees. Due to this independence of the distribution of
thresholds on degrees, as the average degree increases, 
a larger number of neighbors need to be infected before a node
switches to an ``infected" state. For this reason, nodes become
less vulnerable. Since only a single node is infected at the 
beginning, diminishing vulnerability of nodes makes it harder 
for the infection to propagate to a large portion
of the network. 

In contrast, in a network comprising strategic players 
with {\em heterogeneous} degrees, at least for some parameter
regimes, we observe {\em free riding} by nodes with smaller
degrees. A similar free riding is also observed in the 
context of information reliability \cite{Varian}. Interestingly, 
as we show in Section~\ref{sec:MainResults}, when the average 
degree rises, both the fraction of protected population and 
the {\em degree threshold} mentioned in Section \ref{sec:Introduction} 
tend to climb, at least in some parameter space of interest 
(Section~\ref{sec:Numerical}). 

We suspect that the upturn in degree threshold is a consequence of
stronger positive network externalities produced by the investments in 
protection by an {\em increasing} number of higher degree nodes; 
greater positive externalities cause some nodes with larger degrees, 
which would protect themselves when the average degree was smaller, 
to free ride instead. As stated in 
Section~\ref{sec:Introduction}, these unprotected nodes with 
increasing degrees allow an initial infection to propagate 
throughout the network more easily, leading to larger cascade
probability. 

In \cite{La_TON2013}, we carry out a related study with  
some emphasis on cybersecurity. 
However, its model is different from that 
employed here: \cite{La_TON2013} assumes that infections spread 
only to immediate neighbors, whereas the current model allows 
infections to transmit multiple hops. On the other hand, instead of 
binary security choices assumed here, in \cite{La_TON2013} we allow
$M$ ($M \geq 1$) different 
protection levels nodes can select from, in order to capture 
varying cybersecurity measures they can pick.  
Also, an insurer may require a minimum level of protection before 
a node can purchase insurance, and the insurance premium may depend 
on the node's protection level. 

More importantly, besides the differences in their models, there are 
major disparities in the main focus and key findings of these two 
studies. While both studies aim to understand how network 
security is influenced by system parameters, 
we examine in \cite{La_TON2013} network security 
from the {\em viewpoint of a node with a fixed degree}  
as the node degree distribution varies. 
A main finding of \cite{La_TON2013} is that, as the degree 
distribution of neighbors becomes (stochastically) larger, under a set
of assumptions, the average risk seen from neighbors tends to 
{\em diminish} at NEs. In this sense, from the standpoint of
a node with a fixed degree, the network security improves and, 
as a result, the security investments of
nodes with a fixed degree decline. 
Finally, \cite{La_TON2013} also investigates the structural 
relation between an NE and a social optimum that minimizes the 
overall social cost, with the goal of identifying a possible means of 
{\em internalizing} the externalities produced by nodes
\cite{Varian_Microeconomics}.

\section{Model and problem formulation}	\label{sec:Model}

The nodes\footnote{We will use the words {\em nodes} and {\em 
players} interchangeably in the remainder of the manuscript.} 
in a network representing private companies or organizations 
are likely to be interested only in their own objectives. Thus,  
we assume that they are strategic and model 
their interactions as a {\em noncooperative game},
in which players are the nodes in the network. 

We focus on scenarios where the number of nodes is very large.
Unfortunately, 
as stated before, modeling detailed interactions among many nodes 
and analyzing ensuing games is challenging, if possible at all.
A main difficulty is that 
the number of possible strategy profiles we need to consider grows
exponentially with the number of players, and characterizing the NEs 
of games is often demanding even with a modest number 
of players. Moreover, even when the NEs can be computed, it is often 
difficult to draw insight from them. 

For analytical tractability, we employ a {\em population game} model 
\cite{Sandholm}. 
Population games provide a unified framework and tools for studying
{\em strategic interactions} among {\em a large number of agents} 
under following assumptions \cite{Sandholm}. 
First, the choice of an individual agent 
has very little effect on the payoffs of other agents. 
Second, there are finitely many populations of agents, and each 
agent is a member of exactly one population. 
Third, the payoff of each agent depends only on the {\em 
distribution} of actions chosen by members of each population.
In other words, if two agents belonging to the same population
swap their actions, it does not change the payoffs of other
agents. For a detailed discussion of population games, 
we refer an interested reader to the manuscript by 
Sandholm~\cite{Sandholm}.

Our population game model does not capture the microscopic {\em 
edge level} interactions between every pair of neighbors.
Instead, it attempts to capture the {\em mean} behavior of nodes
with varying degrees,   
{\em without} assuming any given network. An advantage of this model 
is that it provides a {\em scalable} model that enables us to study 
the effects of various system parameters on the overall system 
security {\em regardless of} the network size. Moreover, the spirit
behind our population game model is in line with that of Watts'
model \cite{Watts2002} and its extensions (e.g., \cite{Gleeson2007, 
Yagan2012}). 

The notation we adopt throughout the paper is listed in 
Table~\ref{tab:notation}. 

\begin{table}[h]
\begin{center}
 
\begin{tabular}{c|l}
\hline 
$\cA$ & (pure) action space ($\cA = \{I, N, P\}$) \\
$\bC$ & cost function of population game \\
$\bC_{d,a}(\bx)$ & cost of a node from pop. $d$ playing action $a$  \\
$\cD$ & set of node degrees ($\cD = \{1, 2, \ldots, D_{\max}\}$) \\
$D_{\max}$ & maximum degree among nodes \\
$Ins(\bx, d)$ & insurance payout to an insured node from pop. $d$ \\
$K$ & maximum hop distance an infection can propagate \\
$L_P$ & expected loss from an attack for a protected node \\
$L_U$ & expected loss from an attack for a unprotected node \\
$\Delta L$ & $L_U - L_P$ \\
$\bN(\bm, K, \beta_{IA})$ & a Nash equilibrium for given 
	$\bm$, $K$ and $\beta_{IA}$ \\
$c_P$ & cost of protection \\
$c_I$ & insurance premium \\
$d_{\avg}$ & average or mean degree of nodes \\
$d^{NE}$ & degree threshold at a Nash equilibrium \\
$e(\bx)$ & risk exposure at social state $\bx$ \\
$f_d$ or $f_d(\bm)$ & fraction of pop. with degree $d$ \\
$g_{d,a}$ & fraction of pop. $d$ playing action $a$  
	($g_{d,a} = x_{d,a} / m_d$)\\
$\bm$ & pop. size vector ($\bm = (m_d; \ d \in \cD)$) \\
$m_d$ & mass or size of pop. $d$  \\
$p^{i}_{P}$ & prob. of infection for protected nodes  \\
$p^{i}_{U}$ & prob. of infection for unprotected nodes \\
$\Delta p$ & $p^i_{U} - p^i_P$ \\
$w_d$ or $w_d(\bm)$ & weighted fraction of pop. with degree $d$ \\
$\bx$ & social state ($\bx = (\bx_d; \ d \in \cD)$) \\
$\bx_d$ & pop. state of pop. $d$ ($\bx_d = (x_{d, a}; \ a \in \cA)$) \\
$x_{d, a}$ & size of pop. $d$ playing action $a$ \\
$\by^\star(\bm, K, \beta_{IA})$ & a social optimum for given $\bm$, $K$ 
	and $\beta_{IA}$ \\
$\beta_{IA}$ & prob. of indirect attack on a neighbor \\
$\gamma(\bx)$ & prob. that a node will experience an indirect attack \\
	& \myf from a neighbor when the neighbor is attacked \\
$\tau_{DA}$ & prob. that a node experiences a direct attack \\
$\xi_{{\rm cov}}$ & fraction of insurance coverage over deductible \\
\hline
\end{tabular}

\caption{Notation ({\rm pop. $=$ population, prob. $=$ probability}).}
\label{tab:notation}
\end{center}
\end{table}

\subsection{Population game} 	\label{subsec:PG}

We assume that the maximum degree among all nodes is $D_{\max} 
< \infty$. For each $d \in \{1, 2, \ldots, D_{\max}\} 
=: \cD$, population $d$ consists of all nodes with common degree 
$d$.\footnote{Since population $d$, $d \in \cD$, 
comprises all nodes with degree $d$, we also refer to $d$ as 
the degree of population $d$ hereafter. In addition, 
we implicitly assume that there is no isolated node with $d = 0$; 
since isolated nodes do not interact with any other nodes, they are 
of little interest to us.}
We denote the {\em mass} or {\em size} of population $d$ by 
$m_d$, and ${\bf m} := \big( m_d; \ d \in {\cal D} \big)$ is the 
population size vector that tells us the 
sizes of populations with different degrees. Note that $m_d$ does
{\em not} necessarily  represent the {\em number} of agents in 
population $d$; instead, an implicit modeling assumption is
that each population consists of so many agents that a population 
$d \in {\cal D}$ can be approximated as a {\em continuum} of 
{\em mass} or {\em size} $m_d \in (0, \infty)$.\footnote{The
degree-based model we adopt in the study is often known as the 
Chung-Lu model \cite{ChungLu} or the configuration model
\cite{MolloyReed1995, MolloyReed1998}.} 

All players have the same action space ${\cal A} := \{ I, N, P \}$
consisting of three actions -- Insurance ($I$), No Action ($N$),
and Protection ($P$).\footnote{There are other studies where
the investment in security is restricted to a binary case, 
e.g., \cite{BolotLelarge2008, 
KunHeal2003, LelargeBolot2009}. In addition, 
while various insurance contracts may be available on the market 
in practice, as mentioned earlier, since insurance does not affect
the preferences of other players, we believe that the qualitative 
nature of our findings will hold even when different insurance 
contracts are offered.} 
Investment in protection effectively 
reduces potential damages or losses, hence, the {\em risk} for the player. 
In contrast, as mentioned before, insurance simply {\em shifts} 
the risk from the insured to the insurer, without affecting the overall 
societal cost \cite{RA}. For this reason, we focus on
understanding how underlying network properties and 
other system parameters govern the choices 
of players to protect themselves as a function of their degrees
and ensuing social costs. 

{\bf i. Population states and social state --}
We denote by ${\bf x}_d = \big( x_{d, a}; \ a \in \cA \big)$, where
$\sum_{a \in \cA} x_{d, a} = m_d$, the {\em population state} of 
population $d$. The elements $x_{d, a}$, $a \in \cA$, represent
the mass or size of population $d$ which employs action $a$. 
Define ${\bf x} := \big( {\bf x}_d ; \ d \in \cD \big)$ to be the
{\em social state}. Let $\cX_d := \big\{ {\bf x}_d \in \R_+^3 
\ | \  \sum_{a \in \cA} x_{d, a} = m_d \big\}$, where $\R_+ :=
[0, \infty)$, and $\cX := \prod_{d \in \cD} \cX_d$. 
\\ \vspace{-0.1in}

{\bf ii. Costs --} The cost function of the game is denoted by 
${\bf C}: \cX \to \R^{3 D_{\max}}$. For each admissible social 
state ${\bf x} \in \cX$, the cost of a player from population 
$d$ playing action $a \in {\cal A}$ is equal to 
${\bf C}_{d, a}({\bf x})$. In addition to the cost of investing
in protection or purchasing insurance, the costs 
depend on (i) expected losses from attacks and (ii) insurance
coverage when a player is insured. 

In order to explore how network effects and system parameters
determine the preferences of players, we model 
two different types of attacks players experience -- {\em direct} 
and {\em indirect}. While the first type of attacks are not dependent
on the network, the latter depends critically on the underlying 
network and system parameters, 
thereby allowing us to capture the desired {\em network 
effects} on players' choices. 

{\em a) Direct attacks: }
We assume that malicious attacker(s) launch an attack on each 
node with probability $\tau_{DA}$, independently of other 
players.\footnote{Our model can be altered to capture the 
intensity or frequencies of attacks instead, with appropriate
changes to cost functions of the players.} We call
this a {\em direct} attack. When a player experiences a direct
attack, its (expected) cost depends on whether or not 
it is {\em protected}; if the player is protected, 
its cost is given by $L_P$. Otherwise, its cost is equal
to $L_U ( > L_P )$. 

These costs can be interpreted in many different ways. We
take the following interpretation in this paper. 
Assume that each attack leads to a successful
{\em infection} with some probability that depends on the action
chosen by the player. When the player plays $P$, an attack
is successful with probability $p^i_P$, in which case the cost
to the player is given by some rv $C_P$. Otherwise, 
the probability of successful infection is $p^i_U$ and the 
player's cost is given by rv $C_U$, whose distribution may be
different from that of $C_P$. Then, the {\em 
expected} cost due to an infection when attacked is equal to
$p^i_P \cdot \E{C_P}$ when a player is protected and 
$p^i_U \cdot \E{C_U}$ otherwise. 
One can view these expected costs $p^i_P \cdot \E{C_P}$ and 
$p^i_U \cdot
\E{C_U}$ as $L_P$ and $L_U$, respectively, in our model. 
Throughout the paper, we assume $0 \leq p_P^i < p_U^i \leq 1$
and denote the difference $L_U - L_P$ by $\Delta L > 0$. 

{\em b) Indirect attacks: }
Besides the direct attacks by malicious attackers, a player may
also experience {\em indirect} attacks from its neighbors that are
victims of successful attacks and are infected. In order to 
control the manner in which infections spread in the network via 
indirect attacks, we introduce two parameters. First, we assume that an
infected node will launch an indirect attack on
each of its neighbors with probability $\beta_{IA} \in (0, 1]$
independently of each other. We call $\beta_{IA}$ indirect attack 
probability (IAP). Second,
an infection due to a successful {\em direct} attack can propagate
only up to $K \in \N := \{ 1, 2, 3, \ldots \}$ 
hops from its victim.\footnote{This
parameter $K$ can instead be viewed as an average hop 
distance infections spread
with appropriate changes to the cost function.} 
The IAP $\beta_{IA}$ primarily affects the {\em local} 
spreading behavior, whereas the parameter 
$K$ influences how quickly an infection can spread before appropriate 
countermeasures are taken, e.g., a release of patches or vaccines. 
Clearly, as $K$ increases, the infection can potentially spread to a 
larger portion of the network.  

In our model, we assume that the IAP is the same 
whether the infected
node is protected or not, which is reasonable in some cases, e.g., 
spread of computer viruses or worms. But, in some other scenarios, 
this assumption may not hold. For instance, 
in a disease epidemic scenario, e.g., flu, those who are vaccinated
are not only less likely to contract the disease, but also more likely to
recover faster than those who are not vaccinated, 
thereby reducing the odds of transmitting it to others around
them even if they become infected.

Based on the above assumptions, we proceed to derive the cost function 
${\bf C}$ for our population game. Let us denote the mapping 
that yields the {\em degree distribution} of 
populations by ${\bf f}: \R_+^{D_{\max}} \to [0, 1]^{D_{\max}}$, 
where 
\beqan
f_d({\bf m}) = \frac{ m_d }{ \sum_{d' \in \cD} m_{d'} }, 
	\ {\bf m} \in \R_+^{D_{\max}} \mbox{ and } d \in \cD, 
\eeqan 
is the fraction of total population with degree $d$. 
Similarly, define ${\bf w}: \R_+^{D_{\max}} \to [0, 1]^{D_{\max}}$, 
where 
\beqa
w_d({\bf m}) 
\myeq \frac{ d \cdot m_d }{ \sum_{d' \in \cD} d' \cdot m_{d'} }, 
\ {\bf m} \in \R_+^{D_{\max}} \mbox{ and } d \in \cD. 
	\label{eq:bw}
\eeqa
It is clear from the above definition that ${\bf w}$ gives us
the {\em weighted} degree distribution of populations, where
the weights are the degrees. 

Clearly, both ${\bf f}$ and ${\bf w}$ are scale invariant. 
In other words, ${\bf f}({\bf m}) = {\bf f}(\phi \cdot {\bf m})$ 
and ${\bf w}({\bf m}) = {\bf w}( \phi \cdot {\bf m})$ for all 
$\phi > 0$. 
When there is no confusion, we write $\bff$ and ${\bf w}$ in place of
$\bff({\bf m})$ and ${\bf w}({\bf m})$, respectively. 

We explain the role of the mapping ${\bf w}$ briefly. 
Suppose that we fix a social state ${\bf x} \in \cX$ and  
choose a player. The probability that a randomly picked neighbor 
of the chosen player belongs to population $d \in \cD$ is approximately 
$w_d$ because it is proportional to the degree $d$ 
\cite{Callaway, Watts2002}.\footnote{A more careful analysis of the
degree distribution of a randomly selected neighbor is carried out 
in \cite{Makowski2013}, which suggests that it is somewhat
different from what we use here as an approximation. However, 
for large networks without isolated nodes, this discrepancy in 
distributions should be small.}
Hence, the probability that the neighbor has 
degree $d$ and plays action $a \in \cA$ is roughly $w_d \cdot 
x_{d,a} / m_d$. We will use these approximations throughout the
paper. 

Let $\Gamma_k(\bx), \ k \in \N,$ denote the {\em expected}
number of indirect attacks a node, say $i$, experiences through a 
{\em single} neighbor, say $j$, due to successful
{\em direct} attacks on nodes that are $k$ hops away
from node $i$. 
Based on the above observation, we approximate $\Gamma_k(\bx)$
as follows under the assumption that the $K$-hop 
neighborhood of a node can be approximated using
a tree-like 
structure.\footnote{As stated in Section~\ref{sec:Introduction}, 
this assumption may not hold in some of real-world networks as
reported in \cite{Seshadhri2012, WattsStrogatz1998}.
But, we make this assumption to facilitate our analysis. The same 
assumption is introduced in \cite{Gleeson2007, Watts2002, Yagan2012} 
as well.}  For notational
ease, we denote the fraction of population $d \in \cD$ that adopts
action $a \in \cA$ (i.e., $x_{d,a} / m_d$) 
by $g_{d,a}$ and the fraction of unprotected population $d
\in \cD$ (i.e., $(x_{d, N} + x_{d,I}) / m_d = g_{d, N} + g_{d,I}$) 
by $g_{d,U}$ hereafter.  

First, for $k = 1$, 
\beqa
\Gamma_1({\bf x}) 
\myeq \tau_{DA} \cdot \gamma({\bf x}), 
	\label{eq:Gamma_1}
\eeqa
where
\beqa
\gamma({\bf x}) 
\myeq \beta_{IA} \left( \sum_{d \in \cD} w_d  \left( 
 	g_{d,P} \ p_P^i 
	+ g_{d,U} \ p_U^i 
		\right) \right) \lb
\myeq \beta_{IA} \left( p_U^i - \frac{\Delta p}{d_{\avg} \cdot 
	\sum_{d' \in \cD} m_{d'}} 
	\sum_{d \in \cD} d \cdot x_{d,P} \right), 
	\label{eq:gamma}
\eeqa
$\Delta p := p_U^i - p_P^i > 0$, and $d_{\avg} := \sum_{d \in 
\cD} d \cdot f_d$ is the average or mean degree of the 
populations. Note that, from the above assumption,  
$\sum_{d \in \cD} w_d \cdot g_{d,P}$ (resp. $\sum_{d \in \cD}
w_d \cdot g_{d, U}$) 
is the probability that a randomly chosen neighbor is protected (resp.
unprotected). 
By its definition, $\gamma(\bx)$ is the probability that a 
node will see an indirect attack from a (randomly selected) neighbor 
in the event that the neighbor experiences an attack first. 
Similar models have been used extensively in the 
literature (e.g., \cite{Watts2002, Yagan2012}). 
Thus, if the degree of node $i$ 
is $d_i \in \cD$, the expected number of indirect attacks node $i$ 
suffers as a one-hop
neighbor of the victims of successful direct attacks can be 
approximated using $d_i \cdot \Gamma_1(\bx)$. 

Other $\Gamma_k(\bx), k \in \{2, \ldots, K\}$, can be computed in 
an analogous fashion. Suppose that the neighbor $j$ of node
$i$ has degree $d_j \in \cD$. Then, by similar reasoning, 
the expected number of indirect attacks node $j$ suffers as an 
immediate neighbor of the victims of successful direct attacks
other than node $i$ is $(d_j-1) \Gamma_1(\bx)$.
Hence, the expected number of indirect attacks node $i$ sees 
as a two-hop neighbor of the victims of successful direct attacks through 
a single neighbor is given by 
\beqan
&& \myhb \beta_{IA} \left( \sum_{d \in \cD} w_d
	\left( g_{d,P} \ p_P^i 
	+ g_{d,U} \ p_U^i \right)
	\times (d-1) \Gamma_1(\bx) \right) \lb
\myeq \Gamma_1(\bx) \cdot \beta_{IA} \left( \sum_{d \in \cD} w_d (d-1)
	\left( g_{d,P} \ p_P^i 
	+  g_{d,U} \ p_U^i \right)  \right).
\eeqan
Following a similar argument and making use of assumed tree-like
$K$-hop neighborhood structure, we have the following recursive
equation for $k \in \{2, 3, \ldots, K\}$:
\beqa
&& \myhb \Gamma_k(\bx) \lb
\myeq \Gamma_{k-1}(\bx) \cdot \beta_{IA} \left( \sum_{d \in \cD} w_d 
	(d - 1)
	\left( g_{d,P} \ p_P^i 
	+ g_{d, U} \ p_U^i \right) \right) \lb
\myeq \Gamma_{k-1}(\bx) \ \lambda(\bx)
	= \Gamma_1(\bx) \ \lambda(\bx)^{k-1},
	\label{eq:Gamma_k}
\eeqa
where 
\beqa
\myb \lambda(\bx)
& \myb := & \myb \beta_{IA} \left( \sum_{d \in \cD} w_d (d-1)
	\left( g_{d,P} \ p_P^i 
	+ g_{d,U} \ p_U^i \right) \right) \lb
\myeq \beta_{IA} \left( \sum_{d \in \cD} w_d (d-1)
	\left( p_U^i - g_{d,P} \Delta p \right) \right).
	\label{eq:lambda}
\eeqa

Define 
\beqa
e(\bx) 
\myeq \frac{1}{\tau_{DA}} \sum_{k=1}^K \Gamma_k(\bx)
	= \gamma(\bx) \sum_{k=1}^K \lambda(\bx)^{k-1}
	\label{eq:exposure}  \\
\myeq \left\{ \begin{array}{cl}
	\gamma(\bx) \frac{1 - \lambda(\bx)^K}{1 - \lambda(\bx)}
		& \mbox{if } \lambda(\bx) \neq 1, \\
	K \cdot \gamma(\bx) & \mbox{if } \lambda(\bx) = 1. 
	\end{array} \right. 
	\nonumber
\eeqa
We call $e(\bx)$ the (risk) {\em exposure} from a neighbor at social 
state $\bx$. It captures the expected total number of indirect 
attacks a player experiences through a {\em single} (randomly chosen) 
neighbor {\em given} that all nodes suffer a direct
attack with probability one (i.e., $\tau_{DA} = 1$).

We point out two observations regarding the risk exposure. Recall that 
$g_{d, P} = x_{d, P} / m_d$, $ d \in \cD$, denotes the {\em fraction} 
of population $d$ which is protected. 
First, from its definition in (\ref{eq:exposure}) and 
eqs. (\ref{eq:Gamma_1}) - (\ref{eq:lambda}), 
the exposure is determined by $(x_{d, P}; \ d \in \cD)$ 
or, equivalently, $(g_{d,P}; \ d \in \cD)$, without having to know
$(x_{d, I}; \ d \in \cD)$ or $(x_{d, N}; \ d \in \cD)$; each summand
in (\ref{eq:exposure}) can be computed from $\gamma(\bx)$ and 
$\lambda(\bx)$, both of which are determined by $(x_{d, P}; \ d \in \cD)$
or $(g_{d, P}; \ d \in \cD)$ according to (\ref{eq:gamma}) and 
(\ref{eq:lambda}). 
Second, the risk exposure is strictly decreasing in each $x_{d,P}, \ 
d \in \cD$; due to the minus sign in front of $x_{d,P}$ in 
(\ref{eq:gamma}) and $g_{d,P}$ in (\ref{eq:lambda}), 
$\gamma(\bx)$ (resp. $\lambda(\bx)$) is strictly decreasing 
(resp. nonincreasing) in $x_{d,P}$, $d \in \cD$.

We assume that the costs of a player due to multiple 
successful attacks are
additive and that the players are risk neutral.\footnote{While we
assume that the players are risk neutral to simplify the proofs of
our analytical findings in Section~\ref{sec:MainResults}, 
risk aversion can be modeled by altering the cost function and 
similar qualitative findings can be reached at the expense
of more cumbersome proofs; when they are risk averse, we expect the
percentage of population investing in protection or purchasing 
insurance to increase, the extent of which will depend on the level
of risk aversion.} Hence, the expected cost of a player from indirect attacks 
is proportional to $e(\bx)$ and its degree. Based on this
observation, we adopt the following cost function for our 
population game: For any given social state ${\bf x} \in \cX$, the 
cost of a player with degree $d \in \cD$ playing $a \in \cA$ is given 
by 
\beqa
&& \myhb {\bf C}_{d, a}({\bf x}) \lb
\myeq \left\{ \begin{array}{cl} 
	\tau_{DA} \left( 1 + d \cdot e({\bf x}) \right) L_P + c_P
		& \mbox{if } a = P, \\
	\tau_{DA} \left( 1 + d \cdot e({\bf x}) \right) L_U
		& \mbox{if } a = N, \\
	\tau_{DA} \left( 1 + d \cdot e({\bf x}) \right) L_U + c_I
		- Ins({\bf x}, d)
		& \mbox{if } a = I, \\
\end{array} \right. 
	\label{eq:Cost}
\eeqa
where $c_P$ and $c_I$ denote the cost of protection and insurance premium, 
respectively, and $Ins: \cX \times \cD \to \R$ is a mapping that determines
(expected) insurance payout as a function of social state and degree. 
Note that $\tau_{DA} \left( 1 + d \cdot e(\bx) \right)$ is the expected 
number of attacks seen by a node with degree $d$, including both direct and
indirect attacks. 

We assume that the insurance payout for an insured player of degree 
$d \in \cD$ is given by 
\beqa
\ Ins({\bf x}, d)
\myeq \min \big( Cov_{\max}, \ \xi_{{\rm cov}} ({\bf C}_{d,N}({\bf x}) 
	- ded)^+ \big), \lb 
&& \hspace{1.5in} \bx \in \cX, 
	\label{eq:insurance}
\eeqa
where $Cov_{\max}$ is the maximum loss/damage covered by the insurance
policy, $ded$ is the deductible amount, $\xi_{{\rm cov}} \in (0, 1]$ 
is the {\em coverage level}, i.e., the fraction of total damage over 
the deductible amount covered by the 
insurance (up to $Cov_{\max}$), and $(z)^+$ denotes 
$\max(0, z)$. Recall that ${\bf C}_{d,N}(\bx)$ is the cost a node of 
degree $d$ sees from attacks when unprotected. 
As one might expect, the difference in costs
between actions $N$ and $I$ is equal to the insurance 
premium minus the insurance payout, i.e., $c_I - 
Ins({\bf x}, d)$. Moreover, it is clear from (\ref{eq:gamma})
- (\ref{eq:insurance}) that the cost of a player depends on both
its own security level (i.e., protection vs. no protection) and 
those of other players through the exposure $e(\bx)$.

\subsection{Solution concept - Nash equilibria}

We employ a popular solution concept for our study, namely Nash 
equilibria. A social state ${\bf x}^\star$ is an NE if 
it satisfies the condition that, for all $d \in \cD$ and $a \in \cA$, 
\beqa
x_{d,a}^\star > 0 \mbox{ implies } 
{\bf C}_{d,a}({\bf x}^\star) 
\myeq \min_{a' \in \cA} {\bf C}_{d, a'}({\bf x}^\star). 
	\label{eq:NE}
\eeqa
The existence of an NE in a population game is always guaranteed
\cite[Theorem 2.1.1, p. 24]{Sandholm}.

We discuss an important observation that facilitates our study. 
From (\ref{eq:gamma}) - (\ref{eq:insurance}), the cost 
function also has a scale invariance property, i.e., 
${\bf C}(\bx) = {\bf C}(\phi \cdot \bx)$ for all $\phi > 0$. 
This offers a scalable model that permits us to examine the effects 
of various system parameters (e.g., degree distribution of nodes, 
parameter $K$ and 
IAP $\beta_{IA}$) on NEs without suffering from 
the curse of dimensionality even when the population sizes are very 
large: Suppose that ${\cal NE}^\star$ denotes the 
set of NEs for a given population size vector ${\bf m}^1$. Then, 
the set of NEs for another population size vector ${\bf m}^2 = 
\phi \cdot {\bf m}^1$ for some $\phi > 0$ is given by 
$\big\{ \phi \cdot \tilde{\bx} \ | \ \tilde{\bx} \in {\cal NE}^\star \big\}$. 
This in turn means that the set of NEs scaled by the inverse of the 
total population size is the same for all population size vectors with
the identical degree distribution. For this reason, it suffices to study
the NEs for population size vectors whose sum is equal to one, 
i.e., $\sum_{d \in \cD} m_d = 1$. We will make use of this observation
in our analysis in Sections \ref{sec:MainResults} and 
\ref{sec:Numerical}.
\\ \vspace{-0.1in}

\begin{assm} 	\label{assm:0}
We assume that the population size vectors are normalized so that
the total population size is one. 
\\ \vspace{-0.1in}
\end{assm}

Note that Assumption~\ref{assm:0} implies that the population size
vector ${\bf m}$ and its degree distribution $\bff({\bf m})$ are
identical. Hence, a population size
vector also serves as the degree distribution. 
\\ \vspace{-0.1in}

\subsection{Cascades of infection}		\label{subsec:Cascade}

Our model described in the previous subsections aims to capture
the interaction between strategic players in IDS scenarios under the 
assumption that infections typically do not spread 
more than $K$ hops. However, some malwares may disseminate
unnoticed (for example, using so-called zero-day exploits 
\cite{BilgeDumit2012}) or benefit from slow responses by software
developers. When they are allowed to proliferate unhindered 
for an extended period, they may reach a greater
portion of the network than typical infections or malwares can. 
In this subsection, we investigate whether or not such malwares
can spread to a large number of nodes in the network by determining
when {\em cascades} of infection are possible. 

In order to simplify the analysis, we follow an approach similar to
the one employed in 
\cite{Watts2002}. Rather than analyzing a large finite network, we 
consider an infinite network in which the degree
of each node is $d$ with probability $f_d(\bm) = m_d$, $d \in \cD$, 
independently of each other.  By the strong law of large numbers, 
the fraction of nodes with degree $d$ converges to $m_d$ almost
surely for all $d \in \cD$. The elements $x_{d,a}$, $d \in \cD$ and
$a \in \cA$, of the social state $\bx$ can now be interpreted as the 
{\em fraction} of nodes that have degree $d$ and play action $a$. 
Using this model, we look for a condition under which the probability
that the number of infected nodes diverges is strictly positive.
We call this the {\em cascade condition}. 

Fix social state $\bx \in \cX$. 
When there is no confusion, we omit the dependence 
on the social state $\bx$ for notational convenience. 
Suppose that we randomly choose a node, say $i$, 
and then randomly select one of its neighbors, say node $j$. As 
argued in Section~\ref{subsec:PG}, the probability that node $j$ 
is vulnerable, i.e., it will be infected if attacked, is given by
\beqan
\sum_{d \in \cD} w_d \left( g_{d, P} \cdot p_P^i + (1 - g_{d,P}) \ 
	p_U^i \right) = \gamma(\bx) / \beta_{IA}. 
\eeqan

Suppose that we initially infect node $i$, and let the infection 
work its way through the network via indirect attacks (with no 
constraint on $K$). We call the resulting set of all infected nodes
the {\em infected cluster}. When the size of infected cluster is 
infinite, we say that a {\em cascade of infection} took place.

In our model, the nodes are {\em strategic} players and can 
change their actions in response to those of other nodes. 
Therefore, how widely an infection can disseminate starting
with a single infected node, depends on the actions taken by
the nodes at social state $\bx$, 
which are {\em interdependent} via their objectives. 
We are interested in exploring how (a) the probability of 
cascade at NEs and (b) the POA/POS vary 
as we change (i) the node degree distribution, 
(ii) parameter $K$ and (iii) IAP $\beta_{IA}$. 
This study can be carried 
out under assumptions similar to those in \cite{La_TON2013, 
Watts2002, Yagan2012} as explained below.

Following analogous steps as in \cite{Watts2002}, we assume that the 
infected cluster has a tree-like
structure with no cycle. As argued in \cite{Watts2002}, this is a
reasonable approximation when the cluster is sparsely connected.

\begin{figure}[h]
  \centering
  \includegraphics[width=0.35\textwidth]{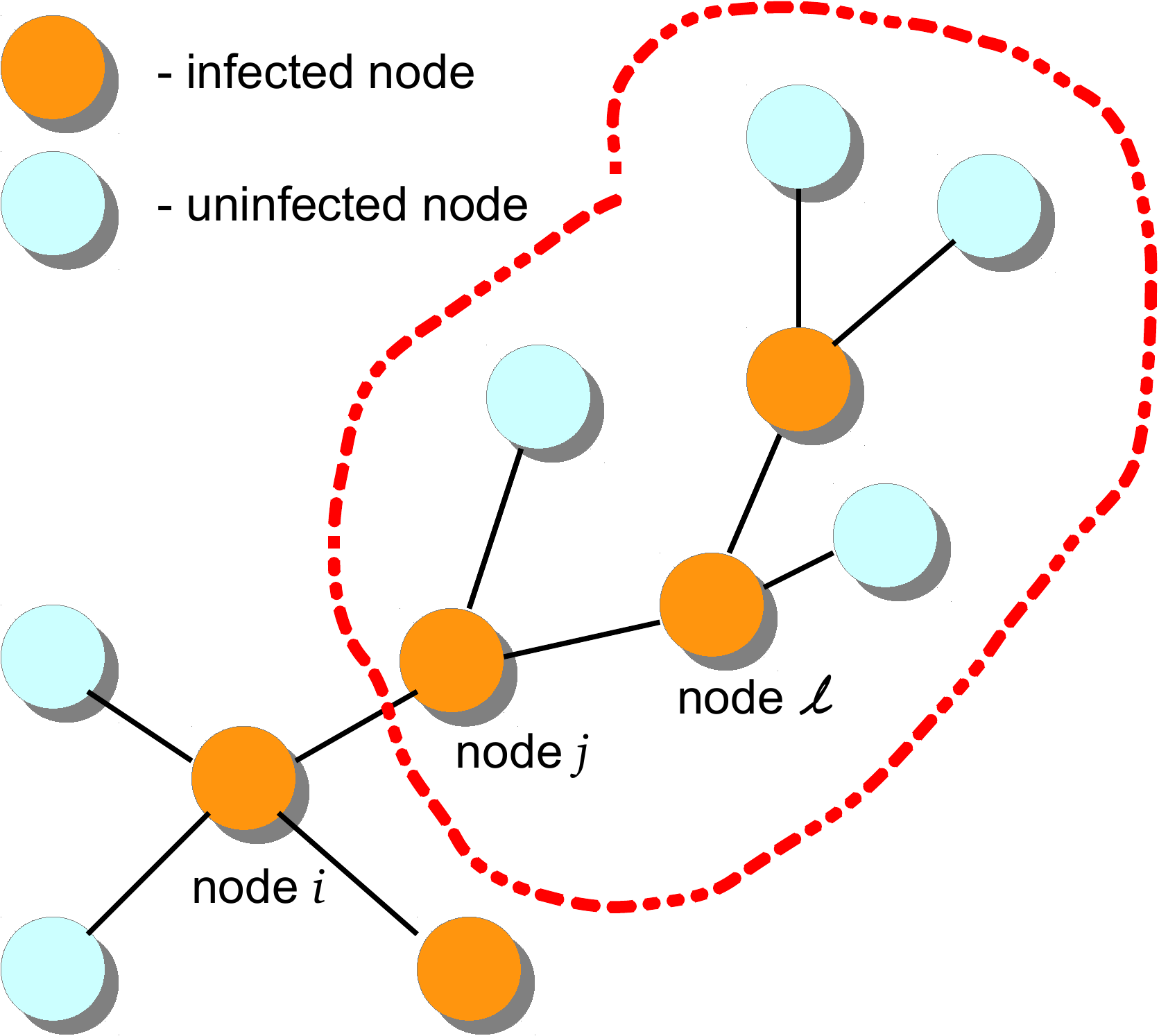}
  \caption{Infected subcluster containing node $j$ and its size $C_j$.}
  \label{fig:cluster}
\end{figure}

Denote the set of node $i$'s neighbors by ${\cal N}_i$. For each $j \in 
{\cal N}_i$, let $C_j$ be the size of the infected subcluster
including node $j$ after removing the remaining cluster connected 
to node $j$ by the edge between nodes $i$ and $j$. An example
is shown in 
Fig.~\ref{fig:cluster}. In the figure, the infected subcluster containing
node $j$ lies inside the dotted red curve.
In this example, $C_j = 3$. 
When neighbor $j$ is not infected, we set $C_j = 0$.
It is clear that a cascade of infection or contagion happens if and only if 
$C_j = \infty$ for some $j \in {\cal N}_i$.

When a neighbor $j$ is infected, the number of $k$-hop neighbors of 
node $j$ in the aforementioned
infected subcluster can be viewed as the size of $k$-th generation in 
Galton-Watson (G-W) model \cite{Grimmett, WatsonGalton}, starting with a single 
individual: Suppose that node $j$ is infected by node $i$ and 
that node $\ell$ is a $k$-hop neighbor of node
$j$ in the infected subcluster that includes node $j$, for some $k \in 
\Z_+$. When $k = 0$, node
$\ell$ is node $j$ itself. Let $N$ denote the number of node $\ell$'s 
infected neighbors that are $k+1$ hops away from node $j$ 
in the same subcluster and contract the 
infection from node $\ell$. 
In the example of Fig.~\ref{fig:cluster}, node $\ell$ is one-hop away from
node $j$ and the rv $N = 1$. 

From its construction, the distribution of $N$ does not depend on 
$k$, and its probability mass function (PMF) ${\bf q}_N : \R \to [0, 1]$ is 
given by
\beqan
&& \myhb {\bf q}_N(n) \lb
\myeq 
\left\{ \begin{array}{l}
	\sum_{d \in \{n+1, \ldots, D_{\max} \} } w^{in}_d {{d-1}\choose{n}} 
		\gamma(\bx)^n (1 - \gamma(\bx))^{d-1-n}  \\
  	\hspace{0.17in} \mbox{ if } n \in \{0, 1, \ldots, D_{\max} - 1\}, \\
	0 \hspace{0.09in} \mbox{ otherwise},
  	\end{array} \right. 
\eeqan
where $\gamma(\bx)$ is the aforementioned probability that the infection of 
a node is transmitted to a neighbor, and ${\bf w}^{in} = (w_d^{in}; \ d \in 
\cD)$ with 
\beqa
w_d^{in} 
\myeq \frac{w_d (g_{d, P} \ p_P^i + g_{d,U} \ p_U^i)}
	{\sum_{d' \in \cD} w_{d'} (g_{d', P} \ p_P^i + g_{d',U} \ p_U^i)} 
	\label{eq:wdin} \\
\myeq \frac{\beta_{IA} \cdot w_d (g_{d, P} \ p_P^i + g_{d,U} \ p_U^i)}
	{\gamma(\bx)}, \ d \in \cD.
	\nonumber
\eeqa
Note that, by definition, $w_d^{in}, 
\ d \in \cD$, is the probability that a neighboring node has degree $d$ 
{\em conditional} on it being a victim of a successful indirect
attack; the numerator of (\ref{eq:wdin}) is the probability that a 
neighbor has a degree $d$ and is vulnerable to infection. 

The number of $k$-hop neighbors in the infected subcluster containing
node $j$, which we denote by $C_j^k, k \in \N$, can now be studied 
using the G-W model. In Fig.~\ref{fig:cluster},  $C_j^k = 1$ for
$k \in \{1, 2\}$ and $C_j^k = 0$ for $k \geq 3$. 
Each individual representing an infected node 
produces $n, \ n \in \{0, 1, \ldots, D_{\max} - 1\}$, offsprings 
according to the PMF ${\bf q}_N$. Consequently, 
the probability $\bP{ C_j < \infty}$ is given by the smallest nonnegative 
root of the equation $Q_N(s) = s$ \cite[p. 173]{Grimmett}, where 
\beqan
Q_N(s) 
\myeq \sum_{n \in \Z_+} {\bf q}_N(n) \ s^n,  \ s \in \R \ \mbox{for which
the sum} \\
&& \hspace{0.9in} \mbox{converges.}
\eeqan
This solution, denoted by $s^\star(\bx)$, always lies in [0, 1]. Moreover, 
$s^\star(\bx) =1$ if (i) $\E{N}
< 1$ or (ii) $\E{N} = 1$ and $q_N(1) \neq 1$. When $\E{N} > 1$,  
we have $s^\star(\bx) < 1$.

By conditioning on the degree of the initial infected 
node, namely node $i$, we obtain
\beqa
&& \myhb \bP{\mbox{cascade takes place at social state } \bx} \lb
\myeq 1 - \sum_{d \in \cD} f_d 
	\big( 1 - \gamma(\bx) (1 - s^\star(\bx)) \big)^d.
	\label{eq:cascade_prob}
\eeqa 
Therefore, assuming $\gamma(\bx) > 0$, $\E{N} > 1$ is a 
sufficient condition for the cascade probability to be strictly positive. 
In addition, except for in uninteresting degenerate cases, 
$\E{N} > 1$ is also a necessary condition. We mention that the task of
determining
whether a cascade is possible or not can be carried out without 
explicitly computing the PMF ${\bf q}_N$ by noting that $\E{N}$ is also 
equal to $\sum_{d \in \cD} w_d^{in} (d-1) \gamma(\bx)$.

Before we proceed, we summarize questions we are interested in exploring 
with help of the population game model described in this section:

\bitem
\item[Q1] Is there a unique NE? If not, what is the structure of 
NEs?

\item[Q2] What is the relation between the degree of a node and its 
equilibrium action? How do the parameters $K$ and $\beta_{IA}$, which
govern the propagation of infections, influence 
the choices of different populations at NEs?

\item[Q3] What is the POA/POS? How do the network properties and 
system parameters affect the POA/POS? 

\item[Q4] How do network properties, in particular the node degree 
distribution and average degree, and system parameters 
shape the resultant probability of cascade at NEs?
\eitem

\section{Main analytical results}	\label{sec:MainResults}

This section aims at providing partial answers to questions Q1 through
Q3 based on analytical findings. 
Before we state our main results, we first state the assumption
we impose throughout this and following sections. 
\\ \vspace{-0.1in}

%

\begin{assm}	\label{assm:1}
The following inequalities hold. 
\begin{enumerate}
\item[a.] $L_P < (1 - \xi_{{\rm cov}}) \ L_U$; and  

\item[b.] $c_P > c_I + ded$. 
\\ \vspace{-0.1in}
\end{enumerate}
\end{assm}

Assumption \ref{assm:1}-a states that when a player is attacked, 
its expected cost is smaller when it is {\em protected} than 
when it is {\em insured}. This implies that the coverage
level is less than 100 percent even when insured. We note that,
in addition to deductibles, coinsurance (i.e., $\xi_{{\rm cov}} < 1$) is 
often used to mitigate the issue of {\em moral hazard} 
\cite{LafMarti} by sharing risk between both the insurer and the 
insured.\footnote{Another way to deal with the issue of 
moral hazard is {\em premium discrimination} that ties 
the insurance premium directly with the security measures adopted
by a player as suggested in \cite{BolotLelarge2008, 
LelargeBolot2009}.} Shetty et al. showed that, in the presence
of informational asymmetry, only a portion of damages would be 
covered by insurance at an equilibrium \cite{Shetty2010b}. Assumption
\ref{assm:1}-b indicates that the investment a player needs to 
make in order to protect itself against possible attacks is
larger than the insurance premium plus the deductible amount. 
We believe that these are reasonable assumptions in many cases. 

We first examine the structure of NEs of the population games
and the effects of parameters $K$ and $\beta_{IA}$ on NEs in Section
\ref{subsec:results-PG}. Then, we investigate the social optimum and
(an upper bound on) the POA/POS as a function of system parameters
in Sections \ref{subsec:SO} and \ref{subsec:POA}, respectively.

\subsection{Population games}	\label{subsec:results-PG}

\begin{theorem} \label{thm:0}
Let ${\bf m} \in \R_+^{D_{\max}}$ be a population size vector and 
${\bf x}^\star \in \cX$ be a corresponding NE for some $K 
\in \N$ and $\beta_{IA} \in (0, 1]$. If $x_{d_1, P}^\star 
> 0$ for some $d_1 \in \{1, 2, \ldots, D_{\max} - 1\}$, then 
$x_{d, P}^\star = m_d$ for all $d > d_1$. 
\\ \vspace{-0.1in}
\end{theorem}

The proof of Theorem~\ref{thm:0} is similar to that of Theorem 
1 in \cite{La_TON2013} and is omitted. 

We note that Theorem~\ref{thm:0} also implies the following: If 
$x_{d_2, P}^\star < m_{d_2}$ for some $d_2 \in \{2, \ldots, D_{\max}\}$,
then $x_{d, P}^\star = 0$ for all $d < d_2$.

In practice, the exposure of a node to indirect attacks will depend on many factors, 
including not only its own degree, but also the degrees and protection levels of 
its neighbors. Therefore, even the nodes with the same degree may behave  
differently. However, one would expect that the nodes with larger degrees 
will likely see higher exposures to indirect attacks and, as a result,
have a stronger incentive to invest in protecting themselves against (indirect) attacks. 
Theorem~\ref{thm:0} captures this intuition. 

The following theorem suggests that, although an NE may not be unique, 
the size of each population $d \in \cD$ investing in 
protection is identical at all NEs. Its proof follows from a straightforward
modification of that of Theorem 2 in \cite{La_TON2013}, and is omitted. 
\\ \vspace{-0.1in}

\begin{theorem}	\label{thm:1}
Suppose that ${\bf x}^1$ and ${\bf x}^2$ are two NEs of the same population 
game. Then, $x^1_{d, P} = x^2_{d, P}$ for all $d \in \cD$. 
\\ \vspace{-0.1in}
\end{theorem}

The uniqueness of the sizes of protected populations at NEs shown in 
Theorem~\ref{thm:1} is crucial for our study. It implies that the 
cascade probability, which we adopt as a (global) measure of network 
security, is identical at all NEs even when there is more than one NE. 
For this reason, it enables us to examine and compare the network 
security measured using cascade probabilities, as we vary the node degree 
distribution and parameters $K$ and $\beta_{IA}$.

Let us explain briefly why an NE is not necessarily unique. 
Suppose that the expected cost of playing $I$ and $N$ is the same 
and is smaller than that of playing $P$ for some population $d$ at an NE. 
Then, there are uncountably many NEs. This is a consequence of an earlier
observation that a purchase of insurance by a player does not affect 
the costs of other players, hence their (optimal) responses. 

Because the populations choosing to protect remain the same at all
NEs (when more than one NE exist) and the issues of interest to us
depend only on populations investing in protection, 
with a little abuse of notation, 
we use $\bN({\bf m}, K, \beta_{IA}) = (\bN_{d,a}({\bf m}, K, \beta_{IA}); 
\ d \in \cD \mbox{ and } 
a \in \cA)$ to denote {\em any} arbitrary NE corresponding to a 
population size vector ${\bf m}$, $K \in \N$ and $\beta_{IA} \in (0, 1]$, 
where $\bN_{d,a}({\bf m}, K, \beta_{IA})$ 
is the size of population $d$ playing action $a$ at the NE. 

Theorems \ref{thm:0} and \ref{thm:1} state that, for fixed
population size vector $\bm$ and parameters $K$ and $\beta_{IA}$, 
there exists a degree threshold given by
\beqan
\ d^{NE}({\bf m}, K, \beta_{IA}) = \min\{ d \in \cD \ | \ 
	\bN_{d, P}({\bf m}, K, \beta_{IA}) > 0 \} 
\eeqan
such that only the populations with degree greater than or equal to 
the threshold would invest in protection at {\em any} NE. 
When the set on the right-hand side (RHS) is empty, we set 
$d^{NE}({\bf m}, K, \beta_{IA}) = D_{\max} + 1$. 
The existence of a degree threshold also greatly simplifies the 
computation of NEs, which are not always easy to compute in general. 

The following theorem sheds some light on how the degree threshold 
$d^{NE}(\bm, K, \beta_{IA})$ behaves with varying $K$ or $\beta_{IA}$. 
\\ \vspace{-0.1in}

\begin{theorem}	\label{thm:2a}
Suppose $K_1, K_2 \in \N$ with $K_1 \leq K_2$. Then, for any population size
vector $\bm$ and IAP $\beta_{IA} \in (0, 1]$, we have 
\beqan
\sum_{d \in \cD} {\bf N}_{d,P}(\bm, K_1, \beta_{IA}) 
\leq \sum_{d \in \cD} {\bf N}_{d,P}(\bm, K_2, \beta_{IA}). 
\eeqan
Similarly, for any population size vector $\bm$ and $K \in \N$, 
\beqan
\sum_{d \in \cD} {\bf N}_{d,P}(\bm, K, \beta_{IA}^1) 
\leq \sum_{d \in \cD} {\bf N}_{d,P}(\bm, K, \beta_{IA}^2) 
\eeqan
if $0 < \beta_{IA}^1 \leq \beta_{IA}^2 \leq 1$. 
\end{theorem}
\begin{proof} 
A proof is given in Appendix~\ref{appen:thm2a}. 
\end{proof}

Theorem~\ref{thm:2a} is quite intuitive; as $K$ or $\beta_{IA}$ increases, 
the effect of a
successful direct attack is felt by a larger portion of the populations. 
Consequently, for any fixed social state $\bx$, the exposure $e(\bx)$ 
grows with $K$ and $\beta_{IA}$. 
As a result, some of population that would not invest
in protection with smaller $K$ or $\beta_{IA}$ 
will see greater benefits of protecting themselves 
because the cost of action $N$ or $I$ increases faster than that of 
$P$ by Assumption~\ref{assm:1}. Consequently, a larger fraction of
population chooses protection. However, as we will show in Section
\ref{subsec:NR2}, these two parameters have very different effects
on the resulting cascade probability.

\subsection{Social optimum}	\label{subsec:SO}

In this subsection, we consider a scenario where there is a single 
social player (SP) that makes the decisions for all 
populations. The goal of the SP is to minimize the overall social cost
given as the sum of (i) damages/losses from attacks and (ii) the cost
of protection. In other words,  
the social cost at social state $\bx \in \cX$ is given by
\beqa
&& \hspace{-0.5in} SC(\bx) \lb
& \hspace{-0.5in} = & \hspace{-0.3in} 
\sum_{d \in \cD} \Big[ \Big( \sum_{a \in \cA} x_{d, a} \cdot 
	{\bf C}_{d, a}(\bx) \Big) + x_{d, I} \big( Ins(\bx, d) - c_I \big) \Big]
	\label{eq:SocialCost0} \\
& \hspace{-0.5in} = & \hspace{-0.3in} 
\sum_{d \in \cD} \Big( x_{d, P} \cdot {\bf C}_{d, P}(\bx)  
	+ (m_d - x_{d,P}) {\bf C}_{d,N}(\bx) \Big). 
	\label{eq:SocialCost1}
\eeqa
Note that $\sum_{d \in \cD} x_{d,I} \left( Ins(\bx, d) - c_I \right)$
in (\ref{eq:SocialCost0}) is the (net) cost for insurer(s). Hence,
the social cost given by (\ref{eq:SocialCost0}) accounts for 
the costs of all players, including the insurer(s). 

Moreover, it is clear from (\ref{eq:SocialCost1}) that the social cost
depends only on $x_{d,P}, \ d \in \cD$, as insurance simply shifts
some of the risk from the insured to the insurer as pointed out 
earlier. For this reason, we can limit the possible atomic
actions of SP to $\{ P, N \}$ and simplify 
the admissible action space of SP to ${\cal Y} := 
\prod_{d \in \cD} [0, m_d]$.  An SP action $\by = \big( y_d; \ d \in 
\cD \big) \in \cY$ specifies the size of each population $d$ that should 
invest in protection (i.e., $y_{d}$) with an understanding that the 
remaining population $m_d - y_d$ plays $N$. 

Let us define a mapping $\bX: \cY \to \cX$, where
\beqan
\bX_{d,a}(\by)
\myeq \left\{ \begin{array}{cl}
	y_{d} & \mbox{if } a = P, \\
	m_d - y_d  & \mbox{if } a = N,  \\
	0 & \mbox{if } a = I. 
	\end{array} \right. 
\eeqan
Fix an SP action $\by \in \cY$. The social cost associated with
$\by$ is given by a mapping $\overline{SC}: \cY \to \R$, where 
\beqa
&& \hspace{-0.5in} \overline{SC}(\by) 
	= SC\big(\bX(\by) \big) \lb
&& \hspace{-0.5in} 
= \sum_{d \in \cD} \Big( y_{d} \cdot {\bf C}_{d, P}\big( \bX(\by)\big)
	+ (m_d - y_d) {\bf C}_{d, N} \big( \bX(\by) \big) \Big). 
	\label{eq:SocialCost2}
\eeqa

The goal of SP is then to solve the following constrained optimization problem. 
\\ \vspace{-0.1in}

\noindent {\bf SP-OPT:} 
\beqa
\min_{\by \in \cY} \ \overline{SC}(\by)
	\label{eq:optim}
\eeqa

Let $\by^\star  \in \arg\min_{\by \in \cY}  \ \overline{SC}(\by)$ denote any minimizer 
of the social cost.
When we wish to make the dependence of $\by^\star$ on the population size
vector $\bm$, parameter $K$ or IAP $\beta_{IA}$ clear, we shall use 
$\by^\star(\bm, K, \beta_{IA})$. 

The following theorem reveals that, like NEs, any minimizer $\by^\star$ has 
a degree threshold so that only the 
populations with degree greater than or equal to the degree threshold should
protect at the social optimum. 

Let $d^\dagger = \min\{ d \in \cD \ | \ y^\star_d > 0 \}$. 
As before, if the set on the RHS is empty, we set 
$d^\dagger = D_{\max} + 1$. 
\\ \vspace{-0.1in}

\begin{theorem}	\label{thm:5}
If $d^\dagger < D_{\max}$, $y^\star_d = m_d$ for all $d > d^\dagger$. 
\end{theorem}
\begin{proof}
A proof is provided in Appendix~\ref{appen:thm5}. 
\end{proof}

We can prove the uniqueness of the minimizer $\by^\star$ by making use of 
Theorem~\ref{thm:5}. 
\\ \vspace{-0.1in}

\begin{theorem}	\label{thm:5a}
There exists a unique solution $\by^\star(\bm, K, \beta_{IA})$ to the SP-OPT problem.
\end{theorem}
\begin{proof}
Please see Appendix~\ref{appen:thm5a} for a proof.
\end{proof}

While the statements of Theorems~\ref{thm:5} and \ref{thm:5a} are
similar to those of Theorems 4 and 5 in \cite{La_TON2013}, the
proofs in \cite{La_TON2013} do not apply to the settings in this
paper.

The following theorem tells us that the protected population size is
never smaller at the social optimum than at an NE. Its proof is similar
to that of Theorem 7 in \cite{La_TON2013} and is omitted.
\\ \vspace{-0.1in} 

\begin{theorem}	\label{thm:7}
Fix a population size vector $\bm$, $K \in \N$ and $\beta_{IA} \in (0, 1]$. 
Let $\bx^\star = \bN(\bm, K, \beta_{IA})$ and $\by^\star = \by^\star(\bm, 
K, \beta_{IA})$. Then, $\sum_{d \in \cD} x^\star_{d, P} 
\leq \sum_{d \in \cD} y^\star_d$. 
\\ \vspace{-0.1in}
\end{theorem}

Theorem~\ref{thm:7} tells us that the damages/losses due to 
attacks are higher at NEs than at the system optimum. 
Hence, because the system optimum is unique, the savings from 
smaller investments in protection at NEs
(compared to system optimum) are outweighed by the increases
in damages.
Thus, the network security degrades as a result of selfish nature of 
the players as suggested in \cite{LelargeBolot2008, LelargeBolot2009}. 
This naturally leads to our next question: How efficient are NEs in 
comparison to the social optimum?

\subsection{Price of anarchy}	\label{subsec:POA}

Inefficiency of NEs is well documented, e.g., \cite{Dubey}, \cite{JohaTsit}, 
\cite[Chap. 17-21]{AGT}, \cite{Rough2005}. In particular, 
the {\em Prisoner's Dilemma}
illustrates this clearly~\cite{Pound}. However, in some cases, the 
inefficiency of NEs can be bounded by finite POA/POS~\cite{KoutPapa}. 

Recall that because all NEs achieve the same social cost in the population 
games we consider by virtue of Theorem~\ref{thm:1}, the POA and POS are 
identical. We are interested in investigating the relation between system 
parameters, including degree distribution 
${\bff}$, $K$ and $\beta_{IA}$, and the POA. 
\\ \vspace{-0.1in}

\begin{theorem}	\label{thm:8}
Let $\bm$ be a population size vector and $d_{\avg}$ be the average degree
of the populations. Suppose $c_P \geq \Delta L \cdot \tau_{DA}$. Then, 
for any $K \in \N$ and $\beta_{IA} \in (0, 1]$,
\beqa
\frac{ SC\big(\bN(\bm, K, \beta_{IA}) \big) }
	{ \tC\big(\by^\star(\bm, K, \beta_{IA})\big) }
	 \leq 1 + d_{\avg} \cdot e_{\max}(\bm, K, \beta_{IA}),  
	\label{eq:thm8}
\eeqa
where
\beqan
e_{\max}(\bm, K, \beta_{IA}) 
\myeq \beta_{IA} \ p_U^i \sum_{k=0}^{K-1} \left( \frac{\beta_{IA} \ p_U^i}{d_{\avg}}
	\sum_{d \in \cD} d (d-1) m_d \right)^k
\eeqan
is the largest possible exposure nodes can see when no population invests in 
protection, i.e., $x_{d, P} = 0$ for all $d \in \cD$. 
\end{theorem}
\begin{proof}
A proof is given in Appendix~\ref{appen:thm8}. 
\end{proof}

The assumption $c_P \geq \Delta L \cdot \tau_{DA}$ in the theorem 
is reasonable because
it merely requires that the insurance premium is at least the 
difference in the expected losses sustained only from a 
{\em direct} attack, not including any additional expected losses a player 
may incur from {\em indirect} attacks. Since a private insurer 
will likely charge a premium high enough to recoup 
the average insurance payout for insured players, the premium 
will need to be higher than $\xi_{{\rm cov}} 
\big( \tau_{DA} (1 + {d}^{ins}_{\avg} \cdot e(\bx))
L_U - ded \big)$, where $d^{ins}_{\avg}$ is the average degree of insured
players. Therefore, assuming that $\xi_{{\rm cov}}$ is not too small and/or the
deductible is not too large, the premium 
is likely to be at least $\Delta L \cdot \tau_{DA}$.

The upper bound on POA in Theorem~\ref{thm:8} is tight in the
sense that there are examples where the POA is equal to the bound. 
We will provide a numerical example in Section \ref{subsec:numerical_POA},
for which the POA is close to our bound in the theorem.

\section{Numerical results}	\label{sec:Numerical}

In this section, we use numerical examples to i) verify our findings in 
Theorems~\ref{thm:0}, \ref{thm:2a} and \ref{thm:8} 
and ii) illustrate how cascade probability is shaped by 
system parameters. For the first three examples in Sections
\ref{subsec:NR1} through \ref{subsec:NR3}, we use a family of (truncated)
power law degree
distributions given by $\big\{ \bm^{\alpha}; \ \alpha \in [0, 3] 
\big\}$, where $m^\alpha_d \propto d^{-\alpha}$, $d \in \cD$.
Over the years, it has been suggested that many of both natural 
and engineered networks have a power law degree distribution 
(e.g., \cite{Albert2000,Lakhina2003}). 
Using Lemma 3 in Appendix D of \cite{La_TON2013}, 
one can easily show that the degree distribution $\bff(\bm^\alpha)$ 
becomes smaller in the usual stochastic order \cite{ShakedShan}
with increasing $\alpha$. 
This implies that the average degree decreases with $\alpha$, which
ranges from 1.33 (for $\alpha = 3$) to 10.5 (for $\alpha = 0$) with
$D_{\max} = 20$. 
For the last example in Section~\ref{subsec:NR4}, we adopt a family
of (truncated) Poisson degree distributions parameterized by $\lambda
\in [1.1, 10.6]$. 

Also, we would like to mention that, although we assume $p_P^i  = 0$
and $p_U^i = 1$ for our numerical examples presented here, similar 
qualitative results hold when other values satisfying $p_P^i < p_U^i$
are used. 

\subsection{Cascade probability, degree threshold, and protected 
population size}	\label{subsec:NR1}

The parameter values used in the first example 
are provided in Table~\ref{table:1}. 

\begin{table}[h]	
\begin{center}
\begin{tabular}{ | c | c || c | c || c | c || }
  \hline                        
  Parameter & Value & Parameter & Value & Parameter & Value \\ \hline
  $\tau_{DA}$ & 0.95 & $p_U^i$ & 1.0 & $p_P^i$ & 0 \\
  $\xi_{{\rm cov}}$ & 0.8 & $ded$ & 20 & $Cov_{\max}$ & 500 \\
  $L_P$ & 10 & $L_U$ & 100 & $\beta_{IA}$ & 0.85 \\ 
  $c_P$ & 300 & $c_I$ & 40 & $D_{\max}$ & 20 \\
  \hline  
\end{tabular}
\end{center}
\caption{Parameter values for first numerical example.}
	\label{table:1}
\end{table}

\begin{figure}[h]
  \centering
  \includegraphics[width=0.24\textwidth]{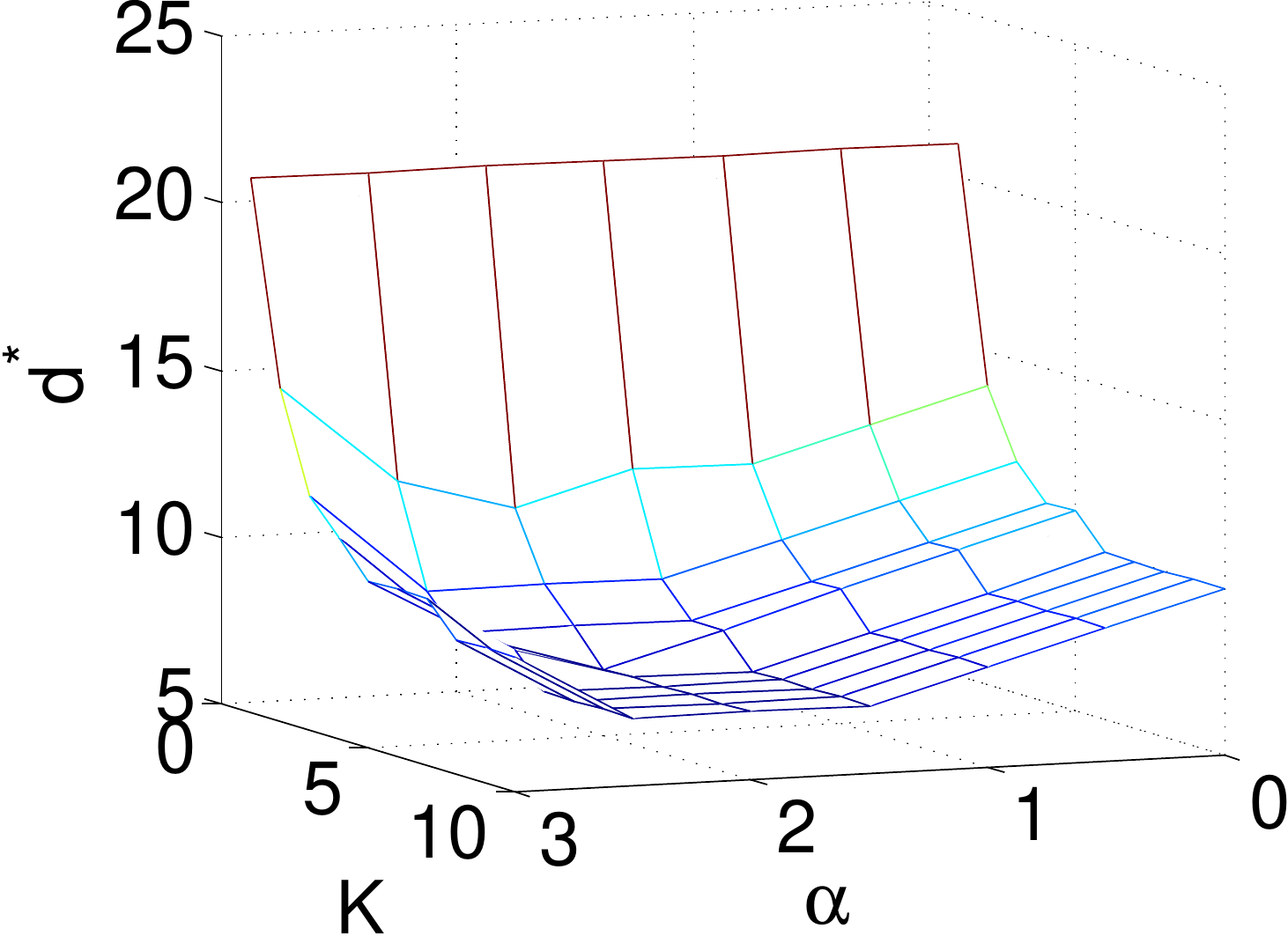}
  \includegraphics[width=0.24\textwidth]{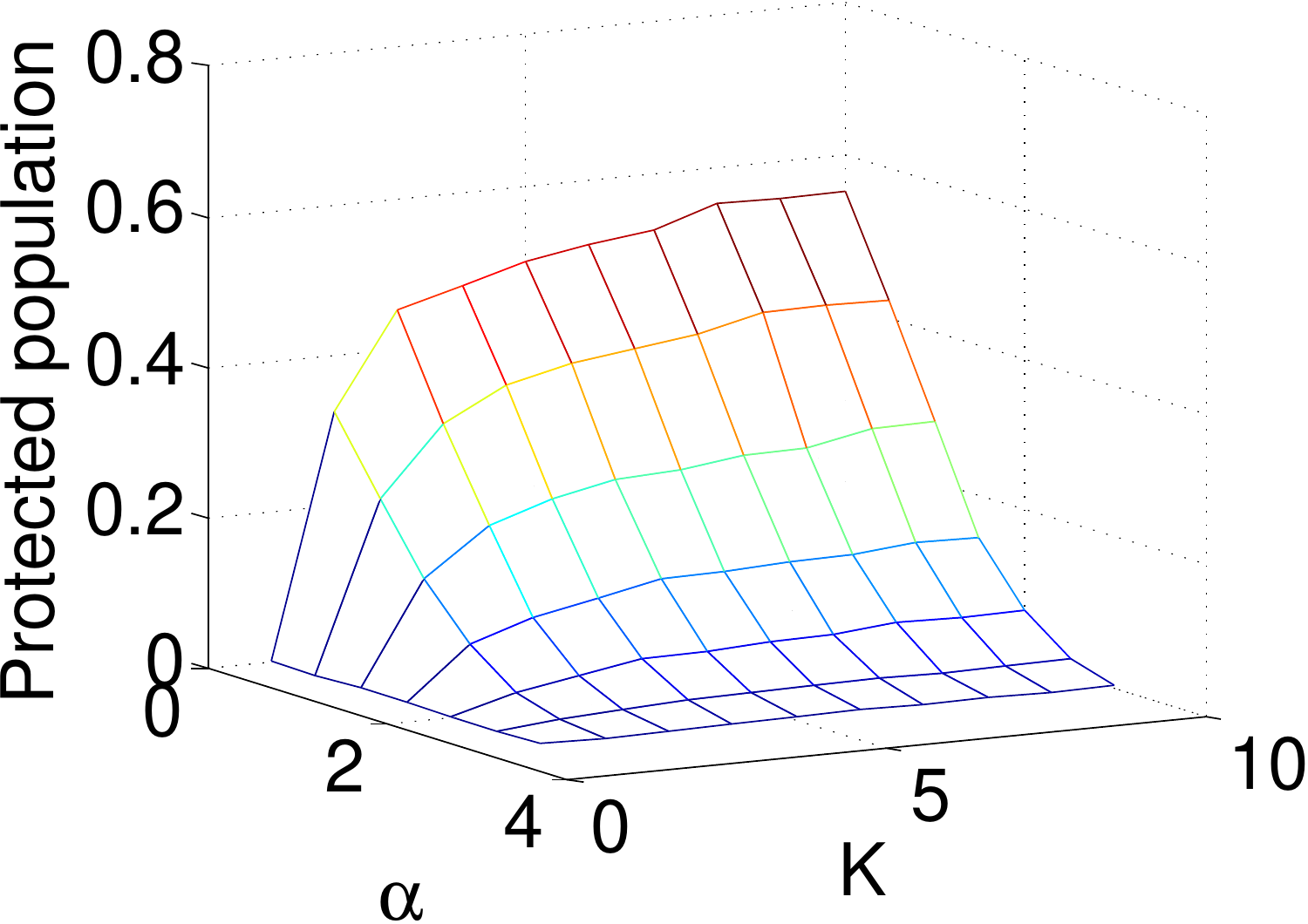}
  \centerline{(a) \hspace{1.4in} (b)}
  \includegraphics[width=0.165\textwidth, angle=90]{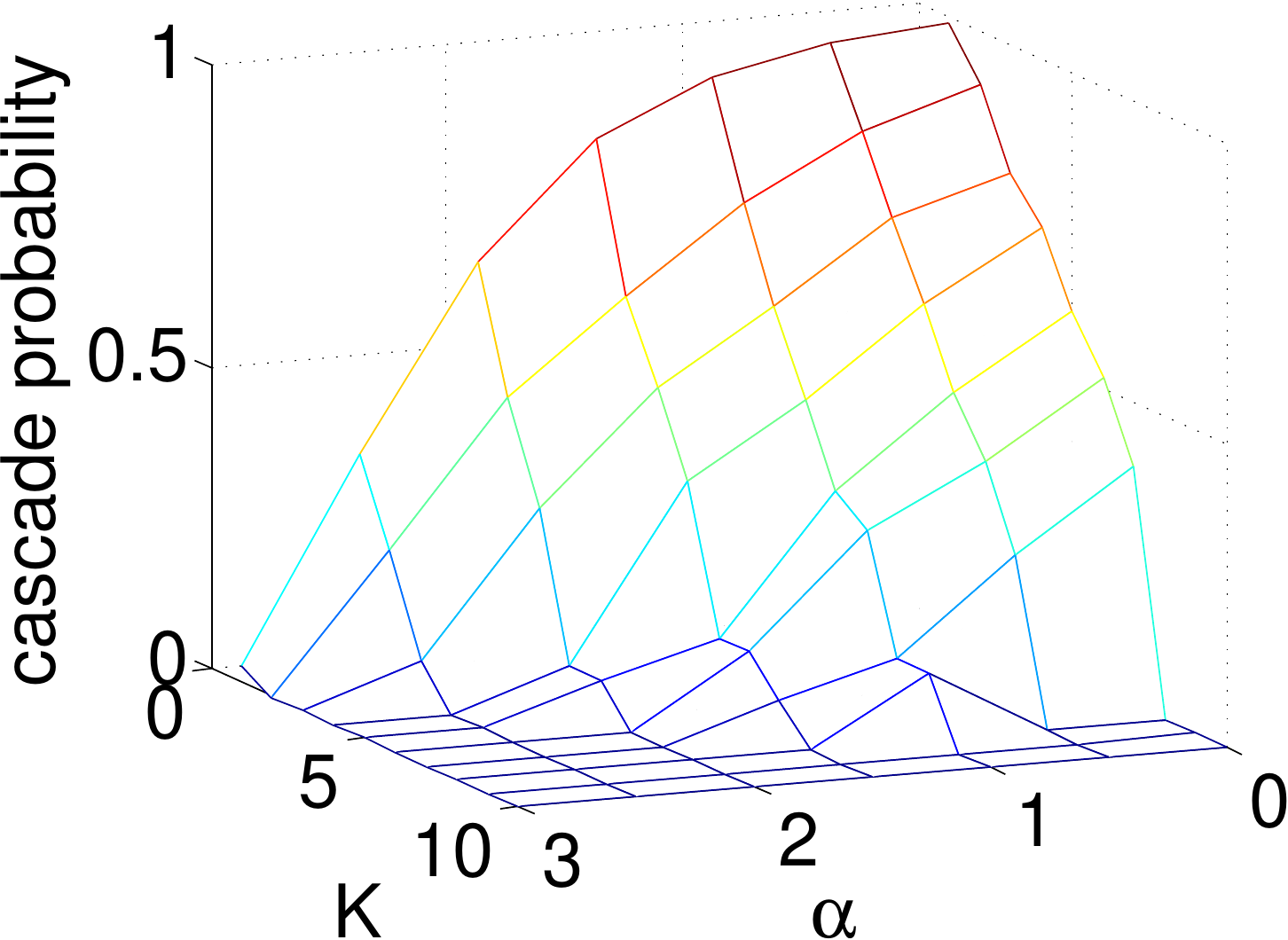}
  \includegraphics[width=0.23\textwidth]{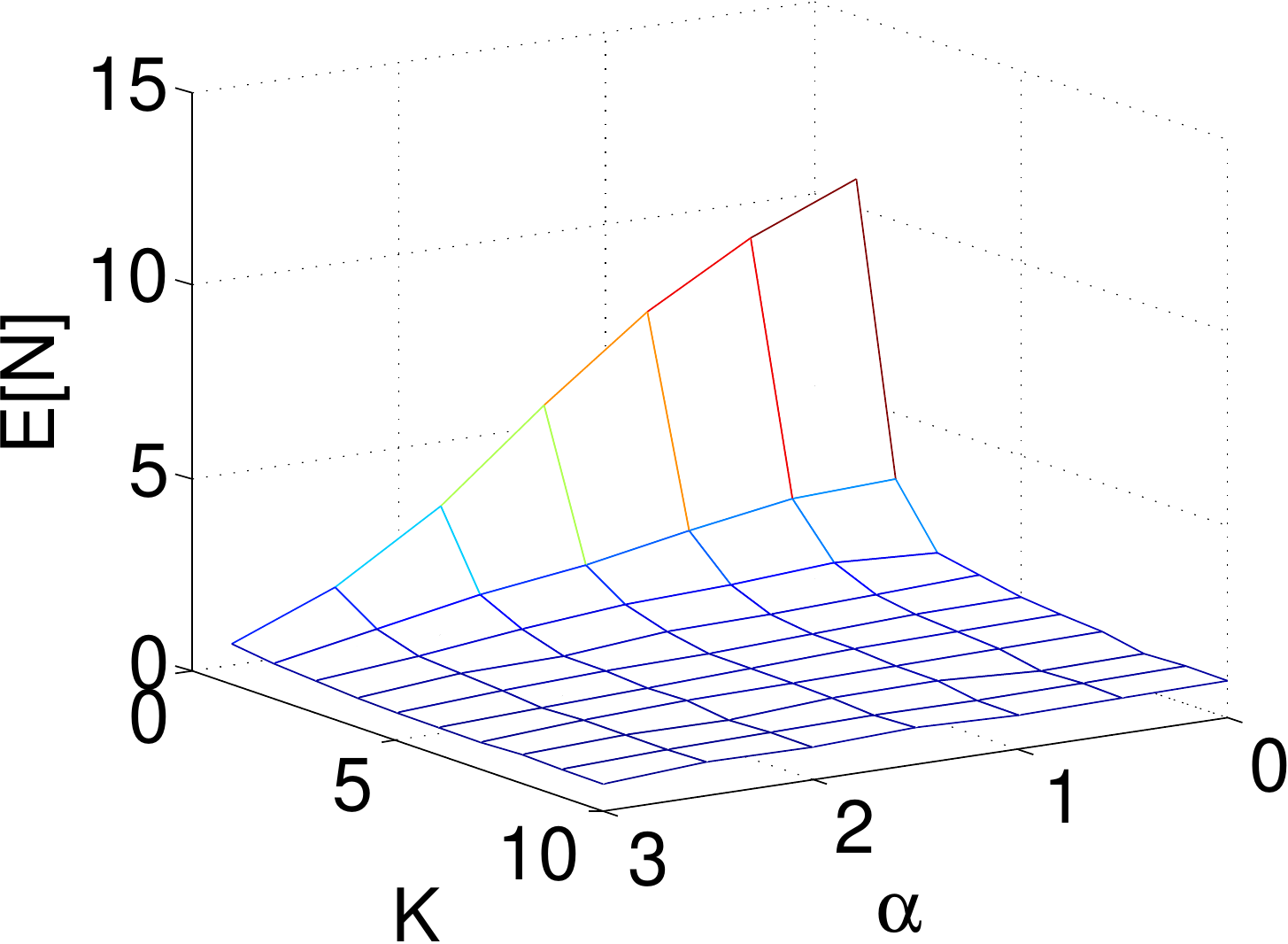}
  \centerline{(c) \hspace{1.4in} (d)}
  \caption{Plot of (a) degree threshold $d^{NE}(\bm^\alpha, K, \beta_{IA})$,
  (b) fraction of protected populations, 
  (c) probability of cascade, and (d) $\E{N}$.}
  \label{fig:1}
\end{figure}

Fig.~\ref{fig:1} plots (a) degree threshold $d^{NE}(\bm^\alpha, 
K, \beta_{IA})$, (b) the fraction of total population that 
invests in protection at NEs, (c) cascade probability given by 
(\ref{eq:cascade_prob}), and (d) $\E{N}$
(discussed in Section~\ref{subsec:Cascade}) 
as a function of the parameter $K$ and 
the power law parameter $\alpha$. It is clear from Fig.~\ref{fig:1}(a)
that, with other parameters fixed, the degree threshold 
$d^{NE}(\bm^\alpha, K, \beta_{IA})$ is nonincreasing in $K$. This leads to a 
larger fraction of total population investing in protection 
(Fig.~\ref{fig:1}(b)) with increasing $K$ 
as proved in Theorem~\ref{thm:2a}. In addition, 
Fig.~\ref{fig:1}(c) shows diminishing cascade probability with 
increasing $K$. 

Fig.~\ref{fig:1}(b) also suggests that the fraction of protected 
population in general goes up with an increasing average degree (or, 
equivalently, decreases with the power law parameter $\alpha$). From
this observation, one might expect the cascade probability to diminish
with the average degree. Surprisingly, Fig.~\ref{fig:1}(c)
indicates that the cascade probability climbs with 
an increasing average degree at the same time. 

We suspect that this somewhat 
counterintuitive observation is a consequence of
what we see in Fig.~\ref{fig:1}(a): Over the parameter settings
where the cascade probability is nonzero, the degree threshold 
$d^{NE}(\bm^\alpha, K, \beta_{IA})$
generally rises with the average degree. This suggests that,
even though more of the population invests in protection,
because nodes with increasing degrees, but smaller than the
degree thresholds are still unprotected and vulnerable, it becomes
easier for an infection 
to propagate throughout the network with the help of such vulnerable 
nodes with increasing degrees.

\subsection{Effects of indirect attack probability $\beta_{IA}$}
	\label{subsec:NR2}

In the second example, we vary IAP 
$\beta_{IA}$ while keeping the values of other
parameters the same as in the first example. Our aim is
to investigate how the IAP influences the 
cascade probability and the fraction of protected population
and compare it to the effects of parameter $K$.  

\begin{figure}[h]
  \centering
  \includegraphics[width=0.24\textwidth]{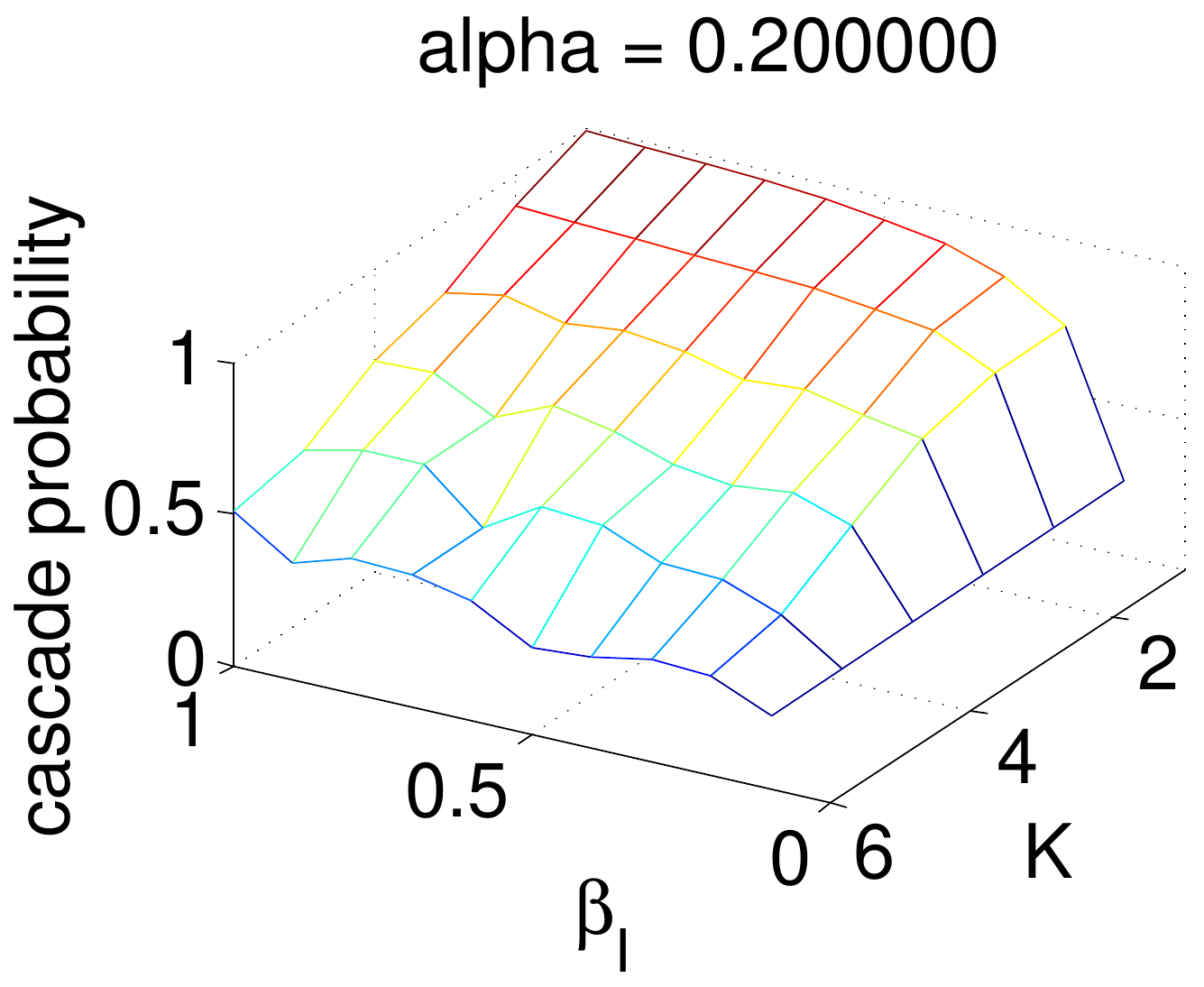}
  \includegraphics[width=0.24\textwidth]{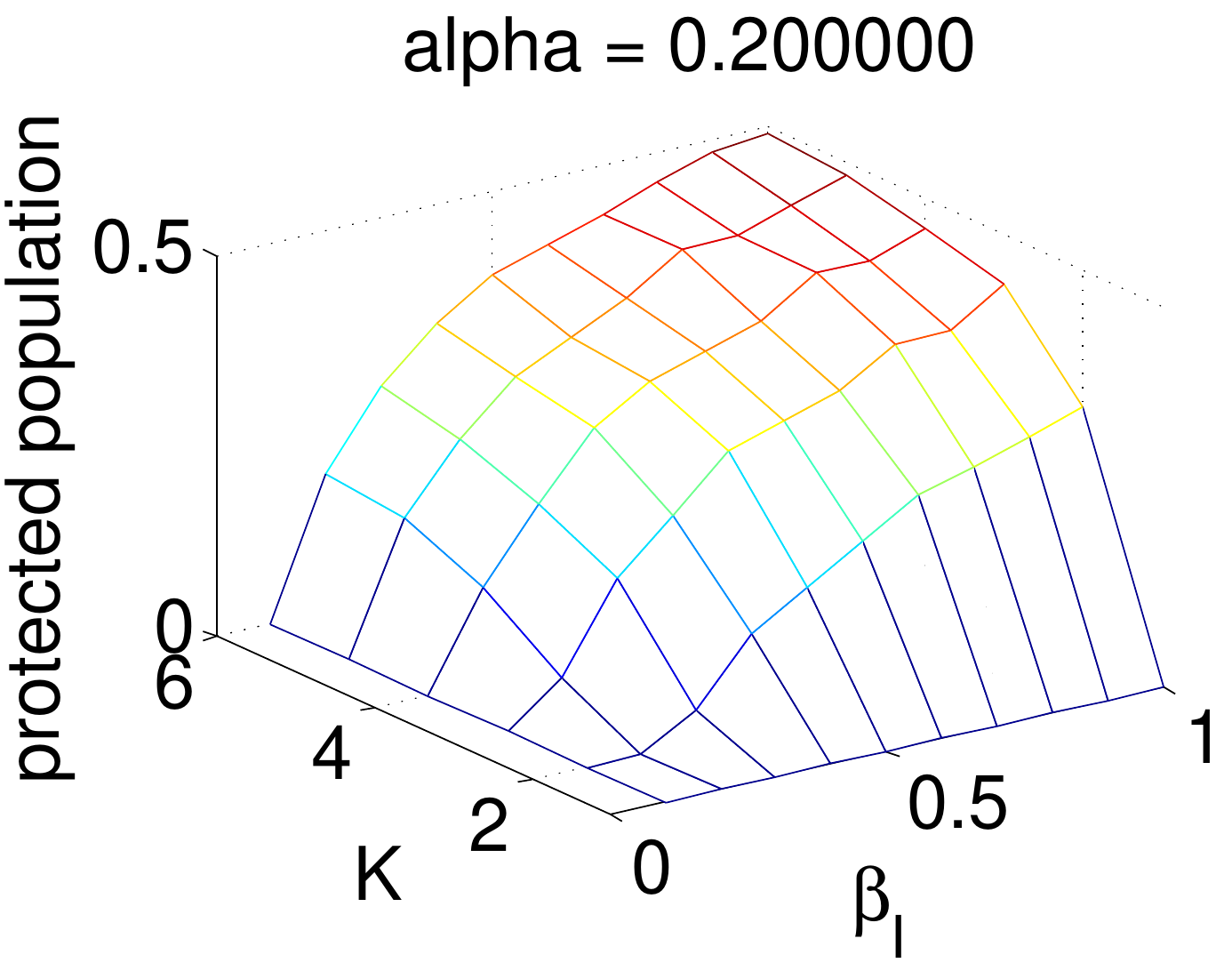}
  \centerline{(a)}
  \includegraphics[width=0.24\textwidth]{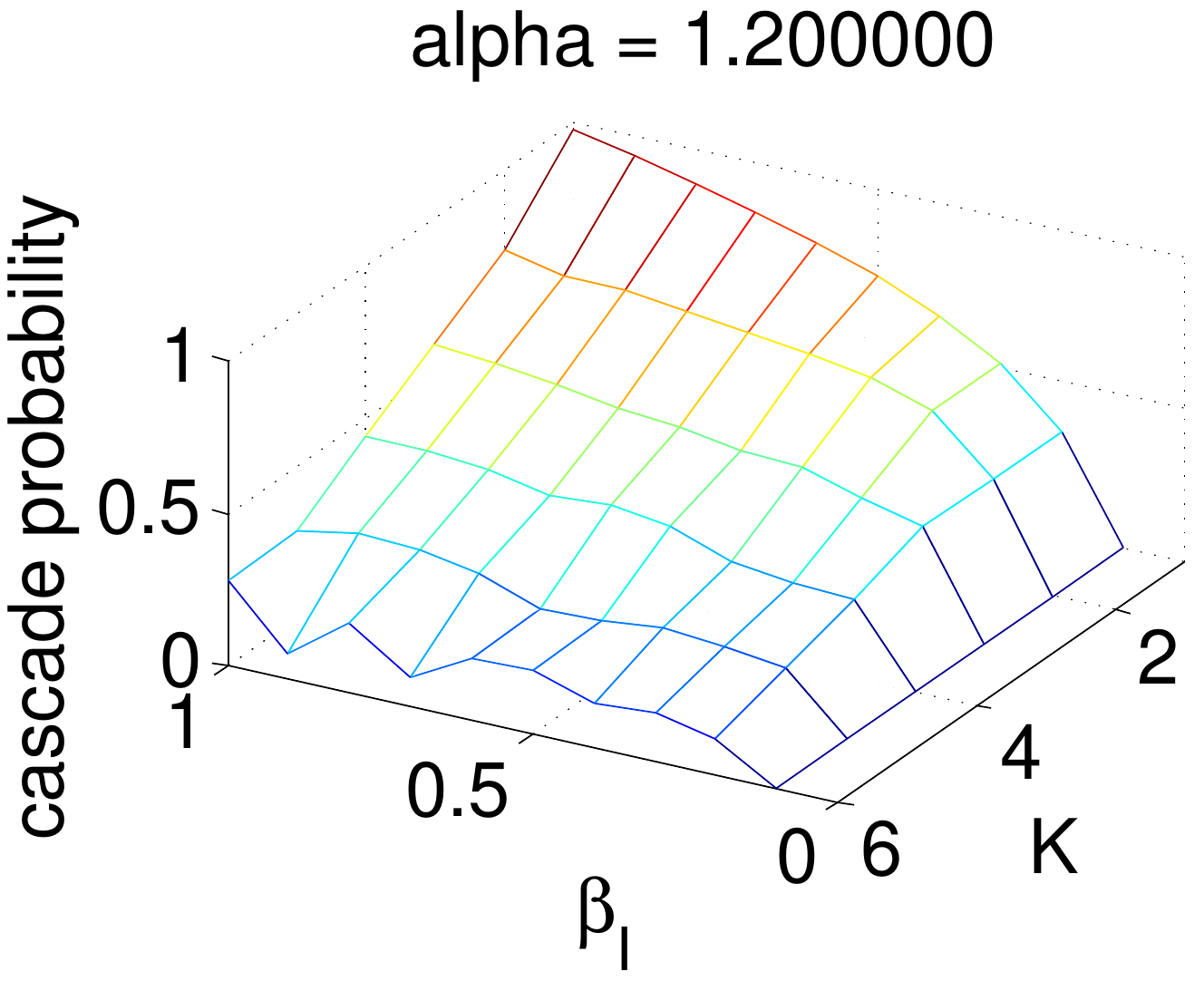}
  \includegraphics[width=0.24\textwidth]{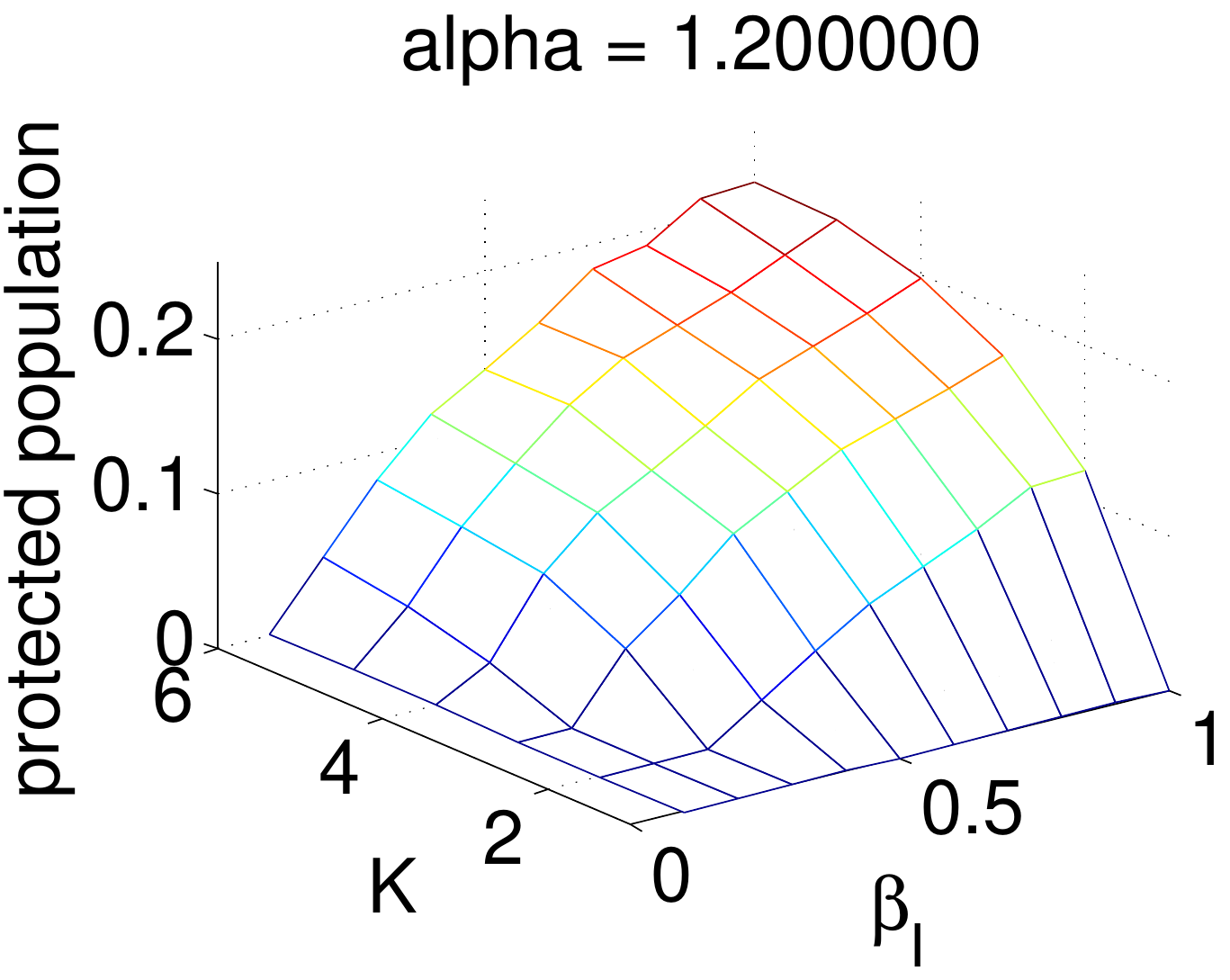}
  \centerline{(b)}
  \includegraphics[width=0.24\textwidth]{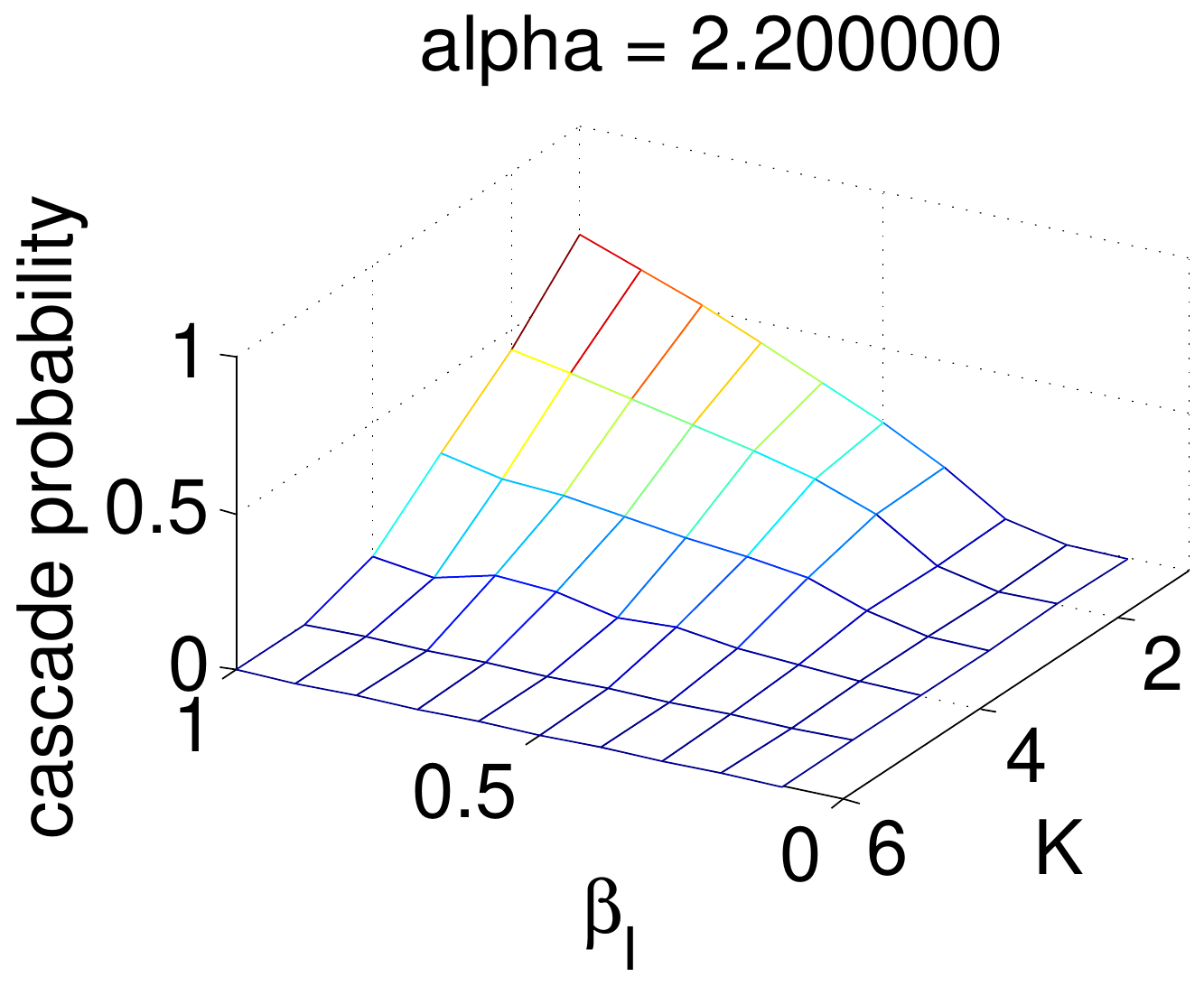}
  \includegraphics[width=0.24\textwidth]{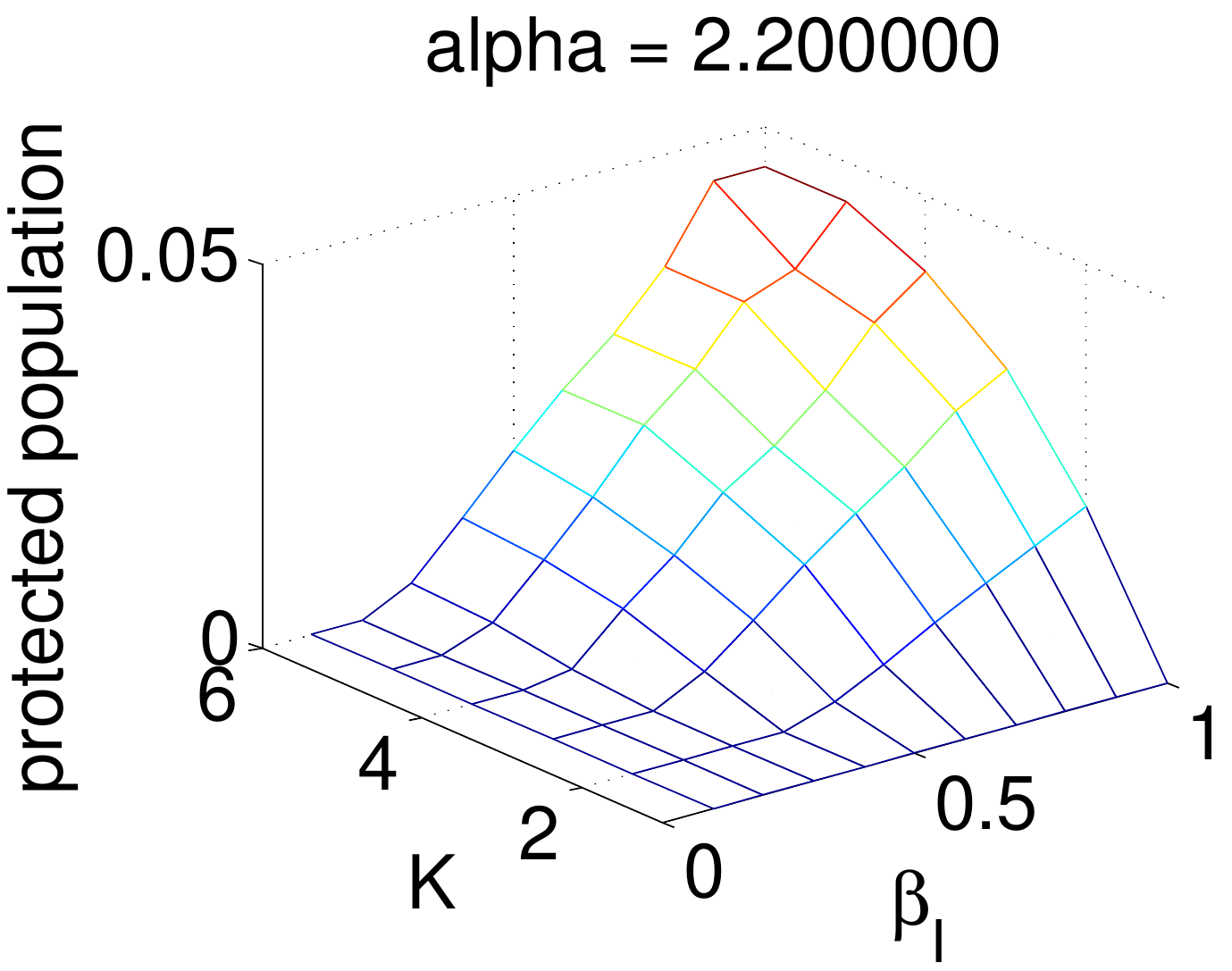}
  \centerline{(c)}
  \caption{Plot of the cascade probability and the protected
  population size. (a) $\alpha = 0.2$, 
  (b) $\alpha = 1.2$, and (c) $\alpha = 2.2$.}
  \label{fig:2}
\end{figure}

Fig.~\ref{fig:2} shows the cascade probability and 
the fraction of population investing in protection as 
IAP $\beta_{IA}$ and parameter $K$
are varied for three different values of $\alpha$ ($\alpha 
=$ 0.2, 1.2 and 2.2). 

We point out three observations. First, as alluded to in the 
first example, the cascade
probability decreases with $\alpha$, as does the fraction of
protected population. As mentioned in Section~\ref{sec:Related}, 
this observation is in sharp contrast with the findings by 
Watts \cite{Watts2002}. 
Figs.~\ref{fig:1} and \ref{fig:2} suggest that when the nodes
are strategic and can choose to protect themselves to reduce
the probability of infection, at least for
certain parameter regimes, the network becomes less 
stable in that the cascade probability rises as the average
degree increases (i.e., $\alpha$ decreases) 
and the second (phase) transition observed in 
\cite{Watts2002} and described in Section
\ref{sec:Related} is missing.

Second, it is clear from Fig.~\ref{fig:2} that, although a 
larger fraction of population invests in protection with 
increasing $\beta_{IA}$ as proved in Theorem~\ref{thm:2a}, the 
cascade probability also rises. What may be surprising at
first sight is how
differently the parameters $K$ and $\beta_{IA}$ affect cascade
probability in spite of the similarity in the way they influence
the portion of protected populations as illustrated in 
Fig.~\ref{fig:2}; 
while raising $K$ results in diminished 
cascade probability, increasing $\beta_{IA}$
leads to rising cascade probability. 

This can be explained as follows: Once other parameters and
social state are fixed, cascade probability does not
depend on $K$. Hence, increasing the protected population size 
reduces cascade probability. On the other hand, with other
parameters and social state fixed, cascade probability climbs
with $\beta_{IA}$. Thus, although the fraction of protected 
population increases with $\beta_{IA}$, because the nodes are 
strategic, they do not invest enough in protection to keep
cascade probability from rising. This can be partially inferred
from growing inefficiency of NEs as hinted by the upper bound
on POA in Theorem~\ref{thm:8}.

Third, Fig.~\ref{fig:2} indicates that the effect of IAP 
is more pronounced when the average degree is larger
in the sense that cascade probability rises more quickly with
the IAP (when it is small).
This is intuitive; when the network
is highly connected, it provides
an infection with a greater number of paths through which the
infection can spread. Hence, even when the IAP
is relatively small, it will be able to propagate
throughout the network more easily.

\subsection{Price of anarchy} \label{subsec:numerical_POA}
	\label{subsec:NR3}

In the next example, we examine the POA as the average degree
of nodes varies. We set $K = 5$ for this example. The values
of other parameters are listed in Table \ref{table:2}. For this
example, we purposely choose parameter values so that the 
POA is close to its upper bound. 

\begin{table}[h]	
\begin{center}
\begin{tabular}{ | c | c || c | c || c | c || }
  \hline                        
  Parameter & Value & Parameter & Value & Parameter & Value \\ \hline
  $\tau_{DA}$ & 0.9 & $p_U^i$ & 1.0 & $p_P^i$ & 0 \\
  $\xi_{{\rm cov}}$ & 0.95 & $ded$ & 5 & $Cov_{\max}$ & 500 \\
  $L_P$ & 5 & $L_U$ & 95 & $\beta_{IA}$ & 0.1 \\ 
  $c_P$ & 88 & $c_I$ & 80 & $D_{\max}$ & 20 \\
  \hline  
\end{tabular}
\end{center}
\caption{Parameter values for first numerical example.}
	\label{table:2}
\end{table}

Fig.~\ref{fig:3} plots (a) degree threshold
$d^{NE}(\bm^\alpha, 5, 0.1)$, (b) fraction of protected 
population, and (c) POA and its bound
in Theorem~\ref{thm:8}. 
We change the $x$-axis to average degree so that it is easier to see 
the effect of average degree on the realized POA and the bound. 
Recall that the average degree decreases with increasing 
$\alpha$.

\begin{figure}[h]
  \centering
  \includegraphics[width=0.155\textwidth]{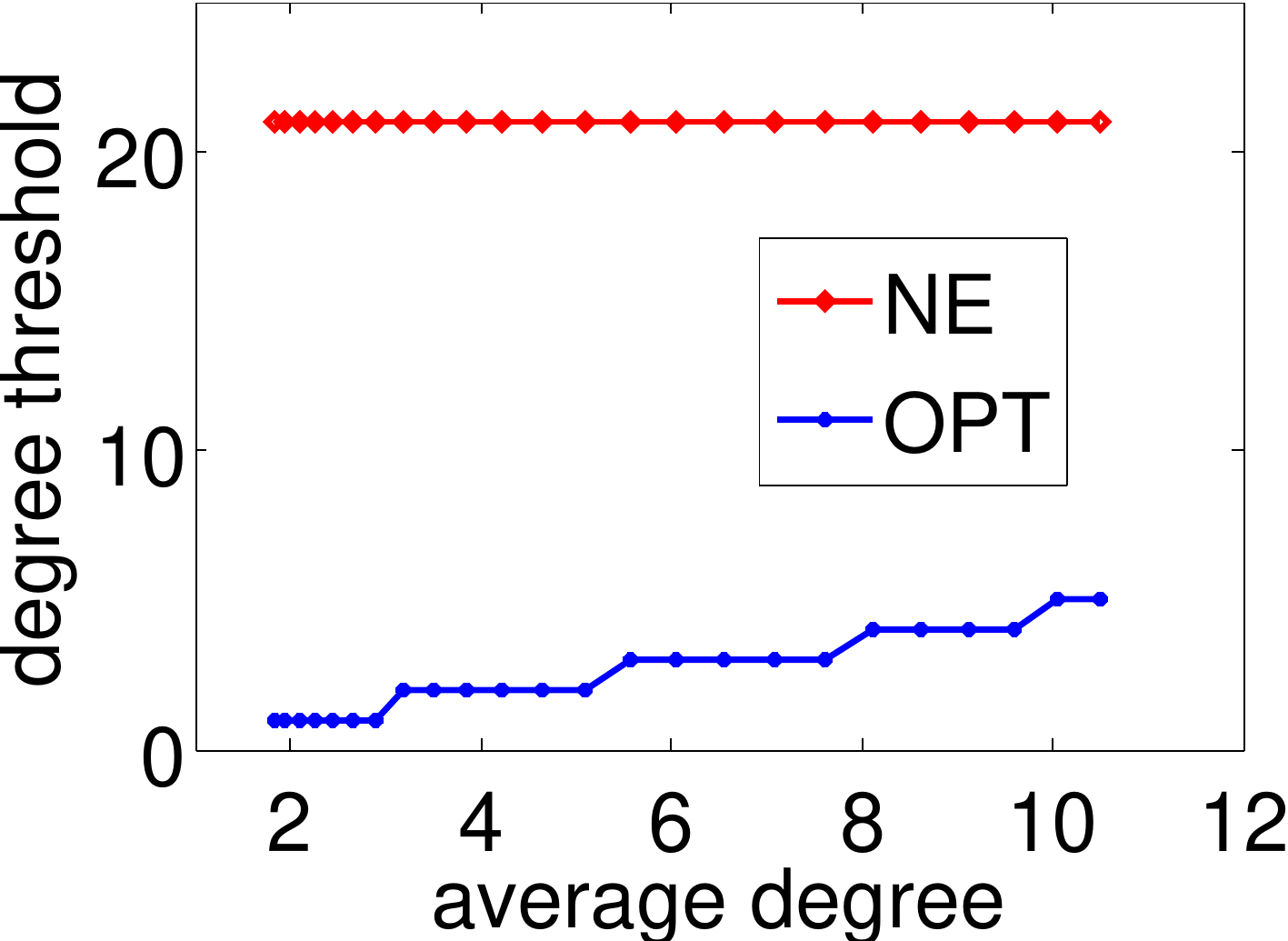}
  \includegraphics[width=0.155\textwidth]{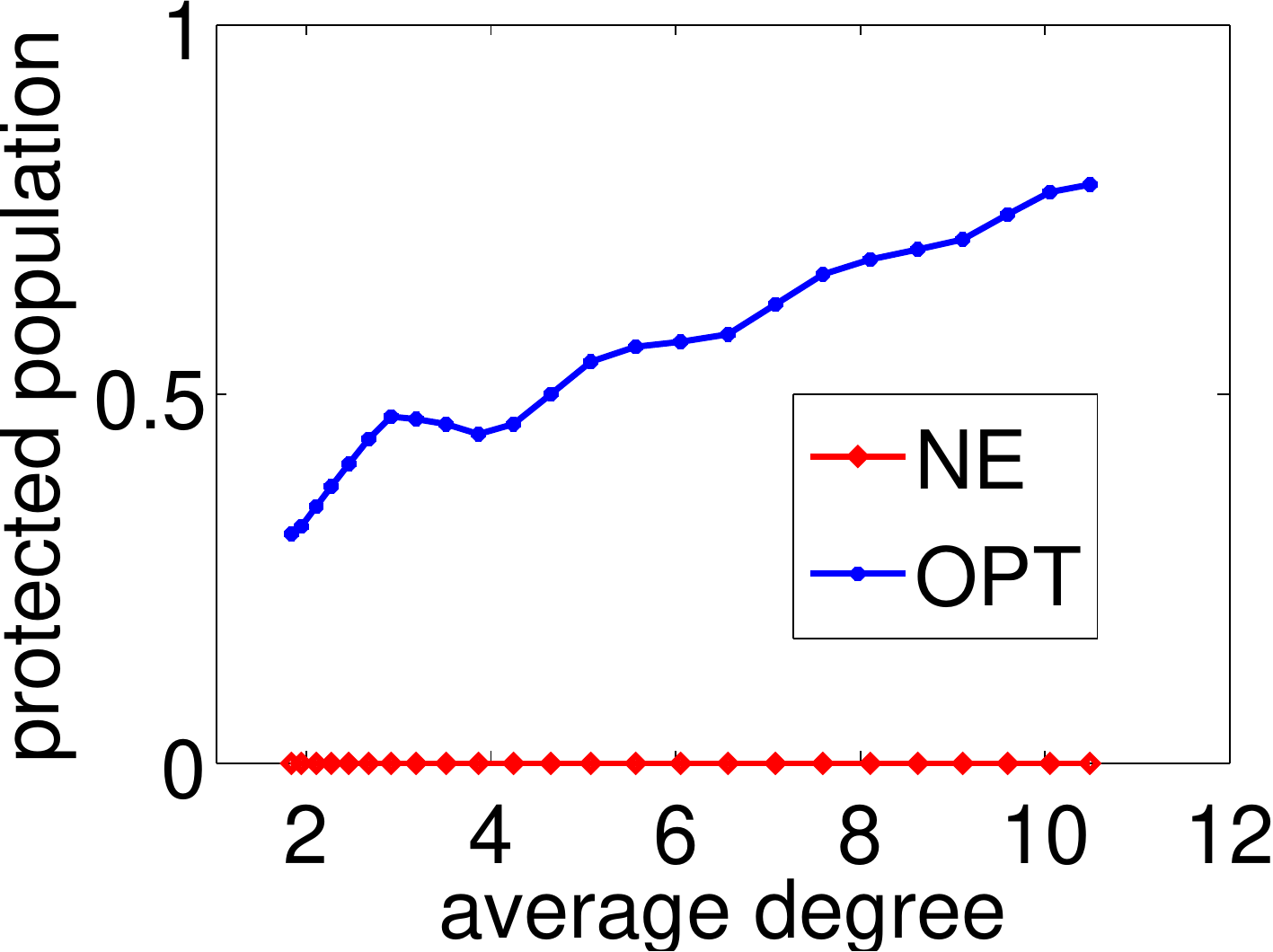}
  \includegraphics[width=0.155\textwidth]{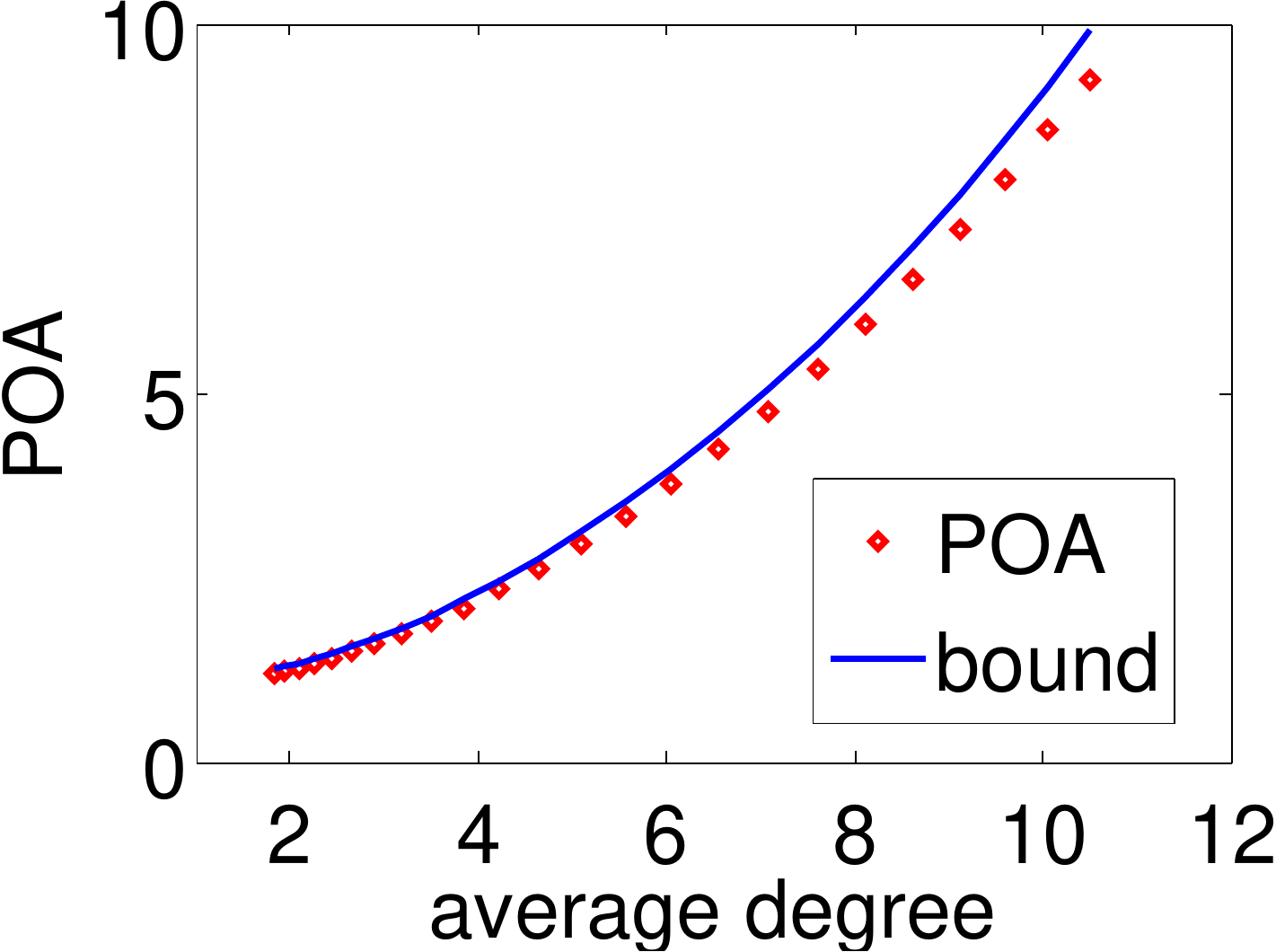}
  \centerline{(a) \hspace{0.9in} (b)\hspace{0.9in} (c)}
  \caption{Plot of (a) degree thresholds, (b) fraction of
  protected population, and (c) POA.}
  \label{fig:3}
\end{figure}

First, it is obvious from Fig.~\ref{fig:3}(c) 
that both the POA and the bound grow much
faster than linearly. Hence, while a greater portion of 
populations elects to protect when $\alpha$ is smaller 
(hence, the average degree is larger) as shown
in previous subsections, the cascade probability rises with an 
increasing average degree, and so does the POA. These findings 
suggest that when the nodes are strategic entities interested only 
in minimizing their own costs, for keeping the cascade probability
small, it is better to have less evenly distributed node degrees 
with fewer large-degree nodes. 

Second, Figs.~\ref{fig:3}(b) and \ref{fig:3}(c) 
tell us the following interesting story. When 
$\alpha$ is small (i.e., the degree distribution is more even), 
although nodes with degrees less than five, which account for about 
20-30 percent of total population, do not invest in protection
at the system optimum, the
POA closely tracks the bound and rises rapidly with the average
degree. Therefore, there is an interesting trade-off one can 
observe: When the degree distribution is less evenly distribution, 
the network is held together by nodes with high degrees. Such
networks are shown to be robust against random attacks, but
are more vulnerable to coordinated attacks targeting high-degree
nodes \cite{Cohen1, Cohen2}. 
One possible way to mitigate the vulnerability is to 
increase the connectivity of the network, hence, the average degree. 
However, Fig.~\ref{fig:3}(a) indicates that increasing network 
connectivity not only leads to higher cascade probability as 
illustrated in previous subsections, but also results in a 
higher social cost and greater POA, which is undesirable.

\subsection{Poisson degree distribution}
	\label{subsec:NR4}

In the last example, we consider a family of (truncated) Poisson 
degree distribution $\big\{ \bm^\lambda; \ \lambda \in [1.1, 10.6] 
\big\}$, where $m^\lambda_d \propto \lambda^d / d!, \ d \in \cD$.
The remaining parameters are 
identical to those in Table~\ref{table:1} of Section~\ref{subsec:NR1}. 
 
\begin{figure}[h]
  \centering
  \includegraphics[width=0.24\textwidth]{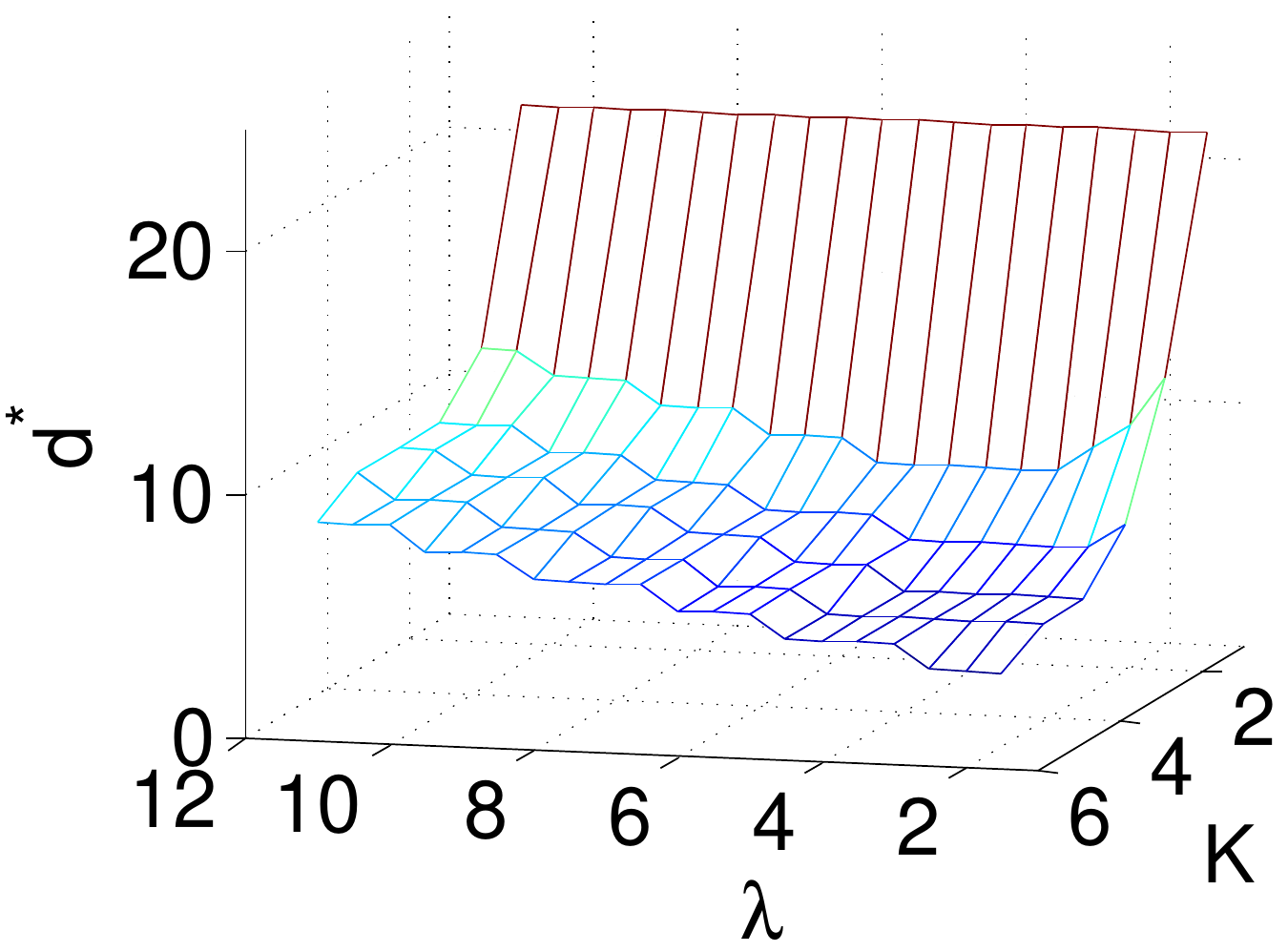}
  \includegraphics[width=0.17\textwidth, angle=90]{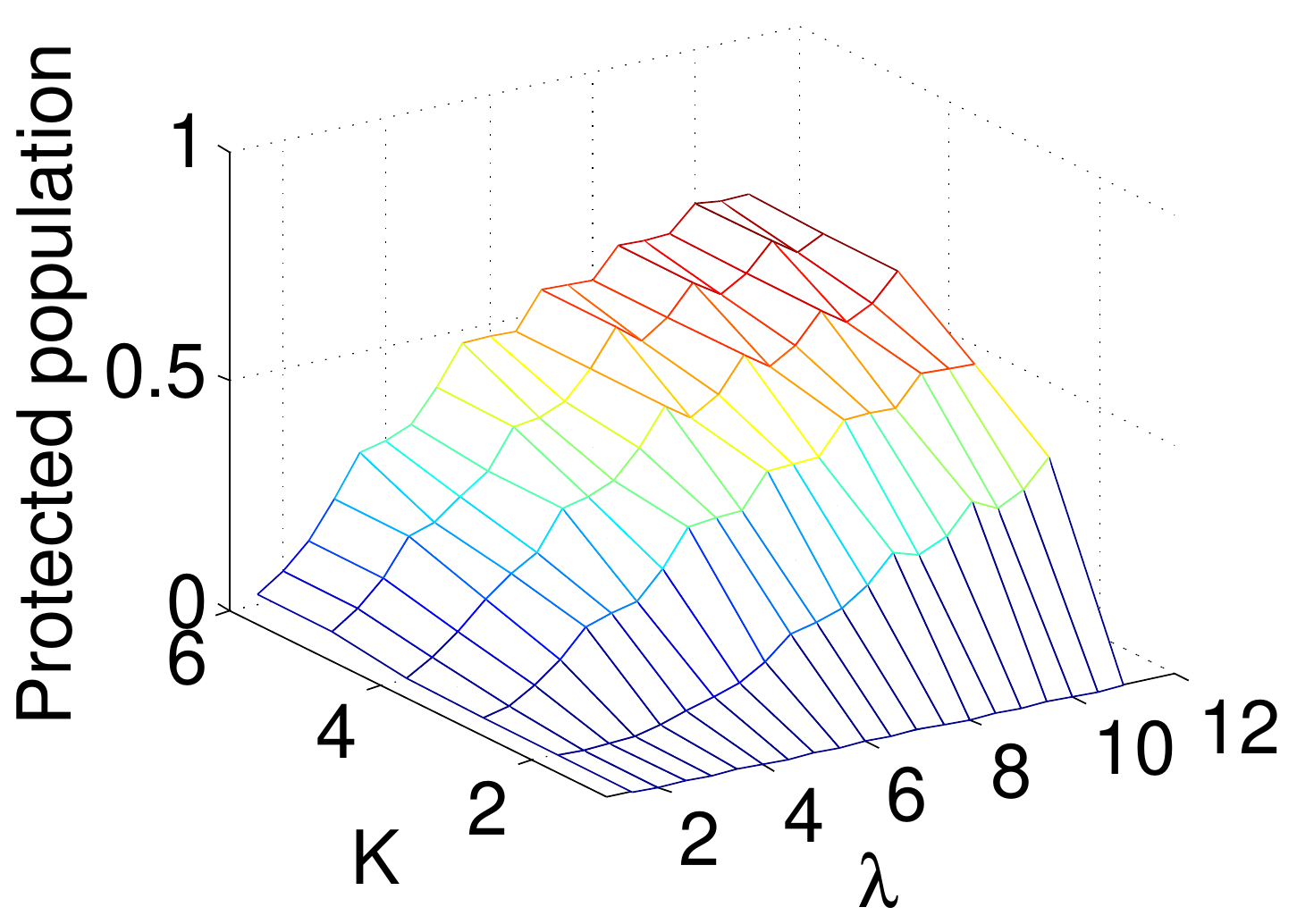}
  \centerline{(a) \hspace{1.4in} (b)}
  \caption{Plot of (a) degree threshold $d^{NE}(\bm^\lambda, 
	K, \beta_{IA})$, and (b) fraction of protected populations.}
  \label{fig:4}
\end{figure}

Fig.~\ref{fig:4} shows (a) the degree threshold $d^{NE}(\bm^\lambda, 
K, \beta_{IA})$ and (b) the fraction of protected populations as a function of
$\lambda$ and $K$. Clearly, the percentage of 
protected population tends to increase with the average degree
(although there is no strict monotonicity), which is consistent with
an earlier observation with power law degree distributions. In addition, 
the degree threshold $d^{NE}(\bm^\lambda, K, \beta_{IA})$ tends to 
climb with the average degree. 

\begin{figure}[h]
  \centering
  \includegraphics[width=0.24\textwidth]{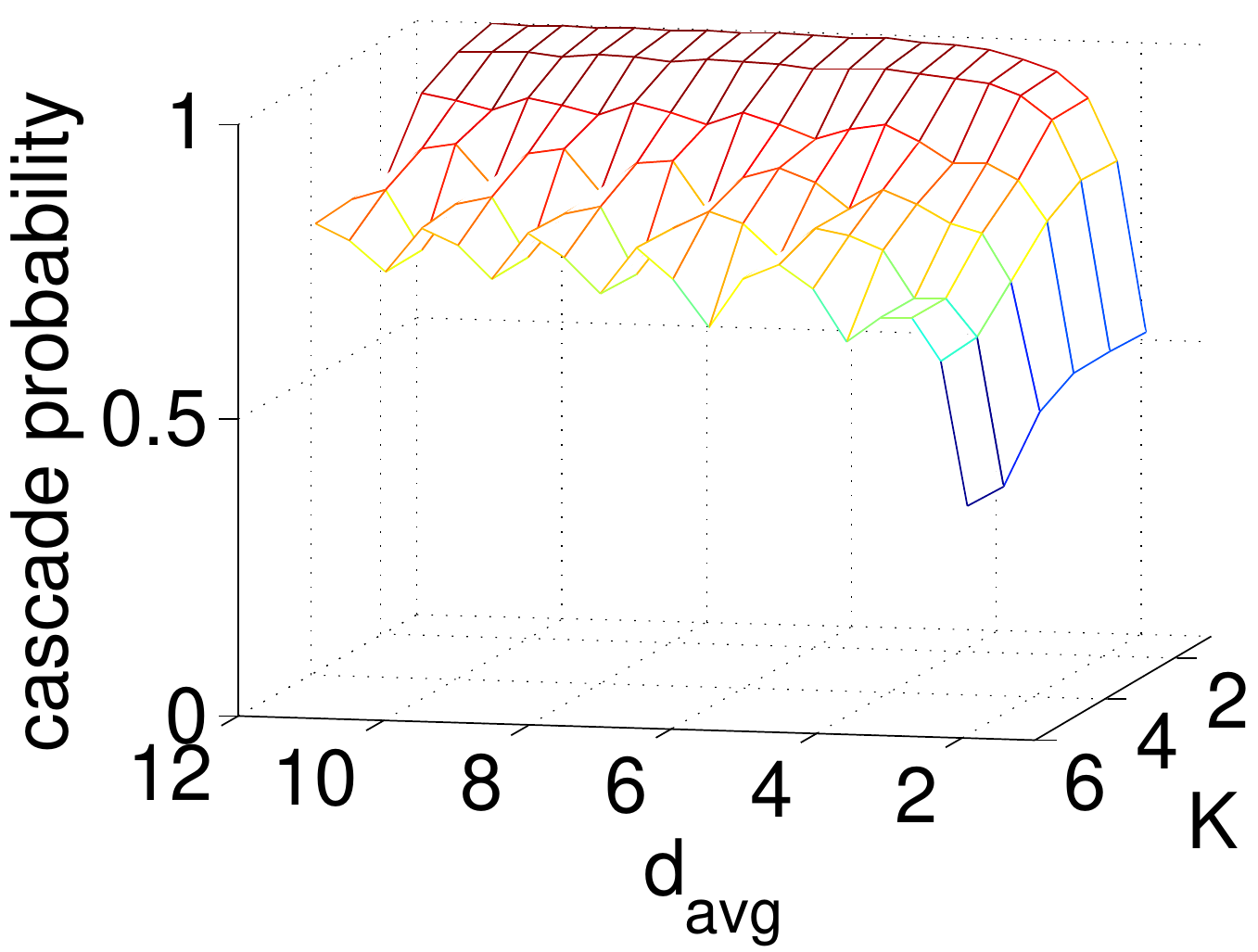}
  \includegraphics[width=0.24\textwidth]{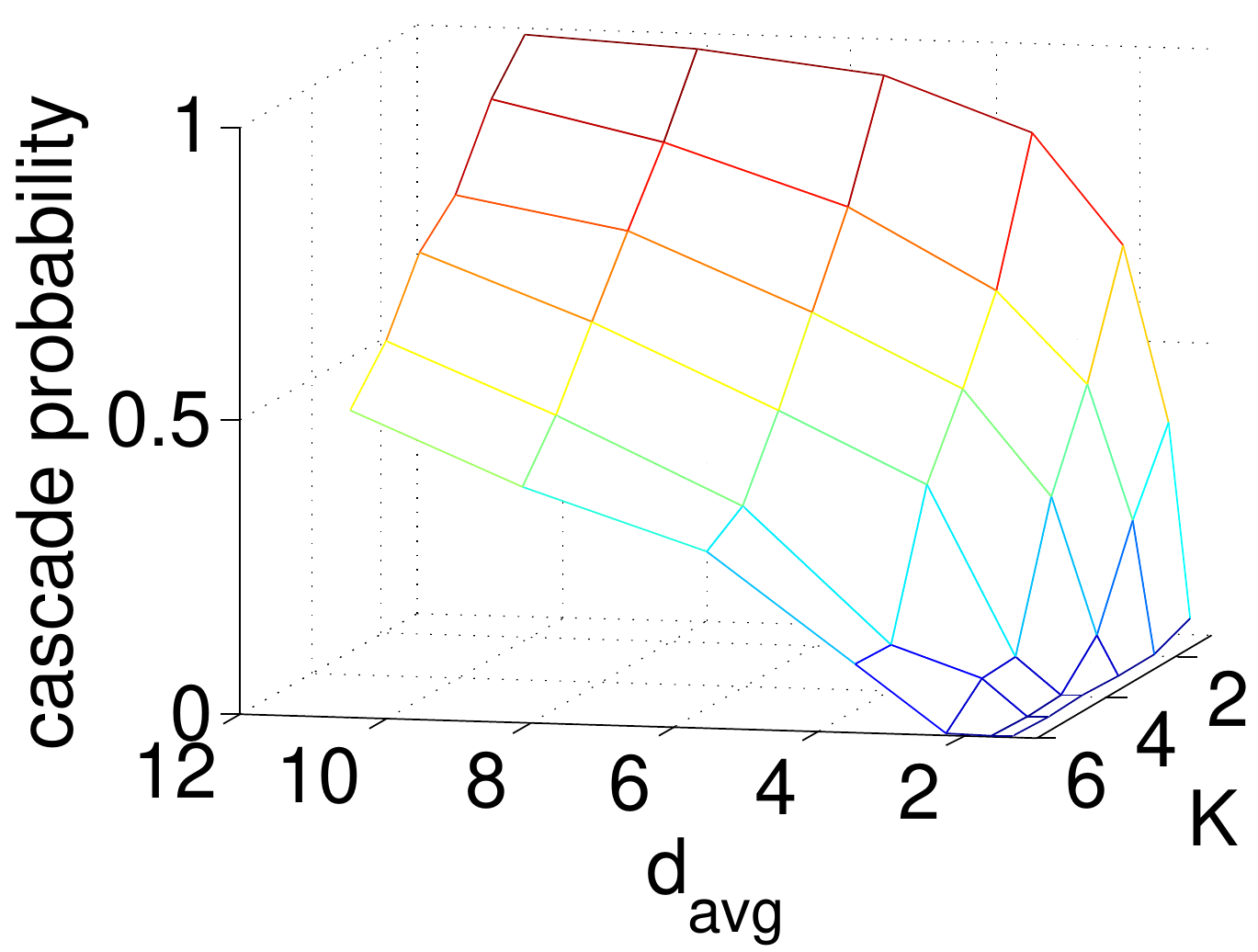}
  \centerline{(a) \hspace{1.4in} (b)}
  \caption{Plot of cascade probability. (a) Poisson distributions, 
  (b) power laws.}
  \label{fig:5}
\end{figure}

In Fig.~\ref{fig:5}, we plot the cascade probability for both (a) Poisson
distributions and (b) power laws. There are two observations we would
like to point out. First, in the case of Poisson distributions, the 
cascade probability shows a cyclic behavior. While similar cyclic patterns
exist with power law degree distributions as well, they are more pronounced
with Poisson degree distributions. These cycles shown in 
Fig.~\ref{fig:5}(a) coincide with the degree threshold
$d^{NE}(\bm^\lambda, K, \beta_{IA})$ in Fig.~\ref{fig:4}(a); the 
dips in the cascade probability 
occur while the degree threshold remains constant. We suspect that 
these cycles are a side effect of a population game model that is 
a deterministic model. 
Nonetheless, the two plots show similar general trends and the
cascade probability reveals an increasing trend with the average degree
for both power law and Poisson degree distributions.

Second, the cascade probability exhibits higher sensitivity with 
respect to the average degree, especially when the average degree is
small, in case of Poisson distributions. In other words, as the average
degree rises, the cascade probability increases more rapidly with
Poisson distributions than with power laws. This suggests that, even for 
a fixed average degree, the cascade probability is likely to depend 
very much on the underlying degree distribution.

\section{Conclusions}	\label{sec:Conclusion}

We studied interdependent security with strategic agents. In 
particular, we examined how various system parameters and network
properties shape the decisions of strategic agents and resulting
system security and social cost. We established the existence
of a degree threshold at both Nash equilibria and social optima. 
Furthermore, we demonstrated the uniqueness of social cost
at Nash equilibria, although there could be more than one 
Nash equilibrium. In addition, we derived an upper bound on 
the POA, which increases superlinearly with the {\em average} 
degree of nodes in general, and demonstrated that the bound is 
tight. Finally, our study suggests that as the average degree 
increases, despite a higher fraction of nodes investing 
in protection at Nash equilibria, cascade probability also rises.


%
%

%
%
%

\appendices

\section{Proof of Theorem~\ref{thm:2a}}	\label{appen:thm2a}

We only prove the first part of the theorem as the second part follows
from essentially an identical argument.
Suppose that the theorem is false for some population size vector $\bm$, 
IAP $\beta_{IA}$,  
and two distinct $K_1$ and $K_2$ satisfying $1 \leq K_1 < K_2$. We will
show that this results in a contradiction. In order to make the dependence
on the parameter $K$ explicit, we denote the cost function ${\bf C}(\bx)$ 
and the exposure $e(\bx)$ by
${\bf C}(\bx ; K)$ and $e(\bx; K)$, respectively. Moreover, for notational
simplicity, we denote ${\bf N}(\bm, K_i, \beta_{IA}), \ i = 1, 2,$ by $\bx^i$ 
in this section. 

We first state a property that will be used shortly.
\\ \vspace{-0.1in}

{\bf {\em Property P1:}} Suppose that $\bar{\bx}^1$ and 
$\bar{\bx}^2$ are two social states such that $e(\bar{\bx}^1) < 
e(\bar{\bx}^2)$. Then, the following inequalities 
hold, which follow directly from the cost function given in 
(\ref{eq:gamma}) - (\ref{eq:insurance}) and Assumption~\ref{assm:1}-a: 
For all $d \in \cD$, 
\beqan
0 < {\bf C}_{d,P}(\bar{\bx}^2) - {\bf C}_{d,P}(\bar{\bx}^1) 
\myl {\bf C}_{d,I}(\bar{\bx}^2) - {\bf C}_{d,I}(\bar{\bx}^1) \lb
\myleq {\bf C}_{d,N}(\bar{\bx}^2) - {\bf C}_{d,N}(\bar{\bx}^1). 
\eeqan
We point out that this property continues to hold even when we compare
social states for two different population sizes $\bm^1$ and
$\bm^2$ or for different values of parameter $K$ or $\beta_{IA}$, 
which satisfy the inequality in the exposures. 
\\ \vspace{-0.1in}

Let $d_i = \min\{ d \in \cD \ | \ x^i_{d,P} > 0\}$, $i = 1, 2$, with an 
understanding $d_i = D_{\max} + 1$ if $x^i_{d,P} = 0$ for all $d \in \cD$. 
Then, from the 
above assumption and Theorem~\ref{thm:0}, we must have $x_{d_1,P}^2
< x_{d_1,P}^1$. By the definition of an NE,
\beqa
{\bf C}_{d_1, P}(\bx^1; K_1) 
\leq \min\{ {\bf C}_{d_1, I}(\bx^1; K_1), \ {\bf C}_{d_1, N}(\bx^1; K_1) \}.
	\label{eq:appen2a-1} 
\eeqa

From the assumption $x_{d,P}^2 \leq x_{d,P}^1$ for all $d \in \cD$
and $x_{d_1, P}^2 < x_{d_1, P}^1$, we get $e(\bx^1; K_1) < 
e(\bx^2; K_2)$. Therefore, property P1 tells us
\beqa
&& \myhb {\bf C}_{d_1, P}(\bx^2; K_2) - {\bf C}_{d_1, P}(\bx^1; K_1) \lb
& \myb < & \myb  \min_{a \in \{N, I\}} \big( {\bf C}_{d_1, a}(\bx^2; K_2) 
	- {\bf C}_{d_1, a}(\bx^1; K_1) \big). 
	\label{eq:appen2a-2}
\eeqa
Together with (\ref{eq:appen2a-1}), the inequality in (\ref{eq:appen2a-2}) 
yields
\beqan
{\bf C}_{d_1, P}(\bx^2; K_2) 
< \min\{ {\bf C}_{d_1, N}(\bx^2; K_2), \ {\bf C}_{d_1, I}(\bx^2; K_2)\}. 
\eeqan
Obviously, this implies $x_{d_1,P}^2 = m_{d_1}$ and, hence, contradicts 
the assumption $x_{d_1,P}^2 < x_{d_1,P}^1 \leq m_{d_1}$.

\section{Proof of Theorem~\ref{thm:5}}	\label{appen:thm5}

We prove the theorem by contradiction. Assume that there exists
$d' > d^\dagger$ such that $y_{d'}^\star < m_{d'}$. 
Suppose that $\epsilon$ is a constant satisfying $0 < \epsilon < 
\min\{ y_{d^\dagger}^\star, \ m_{d'} - y_{d'}^\star \}$ and $\bu_d$ is a 
zero-one vector whose only non-zero element is the $d$-th entry. 
Let
$\by^\dagger = \by^\star + \epsilon \big( \bu_{d'} - \bu_{d^\dagger} \big)$. 
We will show that $\tC(\by^\dagger) < \tC(\by^\star)$,  contradicting 
the assumption that $\by^\star$ is a minimizer of the social cost. 
For notational simplicity, we write $\bx^\dagger$ and $\bx^\star$ 
in place of  
$\bX(\by^\dagger)$ and $\bX(\by^\star)$, respectively, throughout this
section. 

After a little algebra, 
\beqa
&& \myhb \tC(\by^\star) - \tC(\by^\dagger) \lb
\myeq \tau_{DA} \left( e(\bx^\star) - e(\bx^\dagger) \right)
	\sum_{d \in \cD} d \left( y_d^{\star} \ L_P 
		+ (m_d - y_d^{\star}) L_U \right) \lb
&& \myb  + \ \epsilon \ \tau_{DA} (d' - d^\dagger) e(\bx^\dagger) \Delta L.
	\label{eq:appen5-2}
\eeqa
It is clear that the second term in (\ref{eq:appen5-2})
is strictly positive because $d' > d^\dagger$, $e(\bx^\dagger) > 0$ and
$\Delta L > 0$. Thus, to show that (\ref{eq:appen5-2}) is nonnegative, 
it suffices to prove $e(\bx^\dagger) < e(\bx^\star)$. To this end, 
we demonstrate $\gamma(\bx^\dagger) < \gamma(\bx^\star)$ and 
$\lambda(\bx^\dagger) < \lambda(\bx^\star)$. From (\ref{eq:exposure}), 
these inequalities imply $e(\bx^\dagger) < e(\bx^\star)$.

First, from (\ref{eq:gamma}), 
\beqan
\gamma(\bx^\star) - \gamma(\bx^\dagger)
\myeq \frac{ \epsilon \ \beta_{IA} \ \Delta p}{d_{\avg}} \left( d'
	- d^{\dagger} \right) > 0 . 
\eeqan
Second, from (\ref{eq:lambda}), we get
\beqan
\lambda(\bx^\star) - \lambda(\bx^\dagger) 
\myeq \frac{ \epsilon \ \beta_{IA} \ \Delta p}{d_{\avg}} \big( d' (d'-1)
	- d^{\dagger} (d^{\dagger} - 1) \big) > 0 . 
\eeqan

\section{Proof of Theorem~\ref{thm:5a}}	\label{appen:thm5a} 

Suppose that the theorem is not true and there exist two distinct minimizers
$\by^1$ and $\by^2$. By Theorem~\ref{thm:5}, without loss of generality, 
we assume i) $y_d^1 \leq y_d^2$ for all $d \in \cD$ 
and ii) $y_d^1 < y_d^2$ for at least one $d \in \cD$. 
We will show that this leads to a contradiction. Throughout this section, 
we denote ${\bf X}(\by^i), \ i = 1, 2,$ by $\bx^i$ for notational simplicity. 

Let $d^i = \min\{ d \in \cD \ | \ y^i_d > 0\}, \ i = 1, 2,$ with the convention 
$d^i = D_{\max} + 1$ if $\by^i = {\bf 0}$. Note that $d^2 \leq D_{\max}$
by assumption. Since $y^2_{d^2} > 0$, with a little abuse of notation, 
the one-sided partial derivative 
of the social cost with respect to $y_{d^2}$ at $\by^2$ satisfies
\beqan
\frac{\partial}{\partial y_{d^2}} \overline{SC}(\by^2)
:= \lim_{\delta \downarrow 0} \frac{ \overline{SC}(\by^2) - \overline{SC}(\by^2 
	- \delta \cdot \bu_{d^2})}{\delta} \leq 0
\eeqan
with the equality holding 
when $y^2_{d^2} < m_{d^2}$. Define another one-sided
partial derivative of the social cost with respect to $y_{d^2}$ at $\by^1$
to be 
\beqan
\frac{\partial}{\partial y_{d^2}} \overline{SC}(\by^1)
:= \lim_{\delta \downarrow 0} \frac{ \overline{SC}(\by^1 + \delta \cdot \bu_{d^2})) 
	- \overline{SC}(\by^1)}{\delta}. 
\eeqan
Since $\by^1$ minimizes the social cost,
$\partial  \overline{SC}(\by^1) / \partial y_{d^2} \geq 0$. However, 
we will show that $\partial  \overline{SC}(\by^1) / \partial y_{d^2}
< \partial  \overline{SC}(\by^2) / \partial y_{d^2} \leq 0$, leading to 
a contradiction. 

From (\ref{eq:SocialCost2}), 
\beqa
&& \myhb \frac{\partial}{\partial y_{d^2}} \overline{SC}(\by^i)
= {\bf C}_{d^2, P}(\bx^i) - {\bf C}_{d^2, N}(\bx^i) 
	\label{eq:appen5a-1} \\
&& \myhb +  \sum_{d \in \cD} \left( y_d^i \frac{\partial}{\partial y_{d^2}}
	{\bf C}_{d, P}(\bx^i) + (m_d - y_d^i) \frac{\partial}{\partial y_{d^2}}
		{\bf C}_{d, N}(\bx^i) \right). 
	\nonumber
\eeqa
Using the cost function in (\ref{eq:Cost}), we obtain
\beqan
\frac{\partial}{\partial y_{d^2}} {\bf C}_{d, P}(\bx^i)
\myeq \tau_{DA} \ L_P \ d \frac{\partial}{\partial y_{d^2}} e(\bx^i),  \mbox{ and } 
	\lb
\frac{\partial}{\partial y_{d^2}} {\bf C}_{d, N}(\bx^i)
\myeq \tau_{DA} \ L_U \ d \frac{\partial}{\partial y_{d^2}} e(\bx^i).
\eeqan
Here, $\partial e(\bx^i) / \partial y_{d^2}$ and
$\partial {\bf C}_{d, a}(\bx^i) / \partial y_{d^2}$, $a \in \{P, N\}$, are
appropriate one-sided partial derivatives. 
Substituting these in (\ref{eq:appen5a-1}), 
\beqa
&& \myhb \frac{\partial}{\partial y_{d^2}} \overline{SC}(\by^i)
= c_P - \tau_{DA} \left( 1 + d^2 \ e(\bx^i) \right) \Delta L
	\label{eq:appen5a-2}  \\
&& + \sum_{d \in \cD} \tau_{DA} \ d \left( y_d^i \ L_P  + (m_d - y_d^i) L_U \right)
	\frac{\partial}{\partial y_{d^2}} e(\bx^i). 
	\nonumber 
\eeqa

We rewrite $\partial \overline{SC}(\by^1) / \partial y_{d^2}$ in a more
convenient form for our purpose. 
\beqan
&& \myhb \frac{\partial}{\partial y_{d^2}} \overline{SC}(\by^1)
=  c_P - \tau_{DA} \left( 1 + d^2 \ e(\bx^1) \right) \Delta L \lb
&& + \sum_{d \in \cD} \tau_{DA} \ d  \left( y_d^2 \ L_P  + (m_d - y_d^2) L_U 
	+ (y_d^2 - y_d^1) \Delta L \right) \lb
&&  \hspace{0.7in} \times \frac{\partial}{\partial y_{d^2}} e(\bx^1)
\eeqan
Using the above expression, 
\beqa
&& \myhb \frac{\partial}{\partial y_{d^2}} \overline{SC}(\by^2)
	- \frac{\partial}{\partial y_{d^2}} \overline{SC}(\by^1) \lb
\myeq \tau_{DA} \ d^2 \ \Delta L \big( e(\bx^1) - e(\bx^2) \big) \lb
&& - \sum_{d \in \cD} \tau_{DA} \ d  \Big[ \left( y_d^2 \ L_P  + (m_d - y_d^2) L_U  \right) 
	 \label{eq:appen5a-3} \\
&& \hspace{0.7in} \times \left( \frac{\partial}{\partial y_{d^2}} e(\bx^1)
	 - \frac{\partial}{\partial y_{d^2}} e(\bx^2) \right)  \lb
&& \hspace{0.7in} + (y_d^2 - y_d^1) \Delta L \frac{\partial}{\partial y_{d^2}}
	e(\bx^1) \Big]. 
	\nonumber
\eeqa

Because $\by^1 \leq \by^2$, where the inequality is element-wise, and 
$y_{d^2}^2 > y_{d^2}^1$, we have $e(\bx^2) < e(\bx^1)$ and the
first term in (\ref{eq:appen5a-3}) is positive. Moreover, from the 
definition of the exposure in (\ref{eq:exposure}), it is clear 
$\partial e(\bx^1) / \partial y_{d^2} < 0$. Thus, 
in order to prove $(\ref{eq:appen5a-3}) > 0$, it suffices to show 
$\partial e(\bx^1) / \partial y_{d^2} < \partial e(\bx^2) / \partial y_{d^2}$. 
\beqa
&& \myhb \frac{\partial}{\partial y_{d^2}} e(\bx^i) \lb
\myeq \left( \frac{\partial}{\partial y_{d^2}} \gamma(\bx^i) \right)
	\sum_{k=0}^{K-1} \lambda(\bx^i)^k 
	+ \gamma(\bx^i) \sum_{k=0}^{K-1} \frac{\partial}{\partial y_{d^2}} 
		 \lambda(\bx^i)^k \lb
\myeq - \frac{\beta_{IA} \ \Delta p \ d^2}{d_{\avg}} \left[ 
	\sum_{k=0}^{K-1} \lambda(\bx^i)^k \right. 
	\label{eq:appen5a-4} \\
&& \hspace{0.7in} \left. + \gamma(\bx^i) (d^2 - 1)
	\sum_{k=1}^{K-1} k \ \lambda(\bx^i)^{k-1}  \right].
	\nonumber
\eeqa
As $\gamma(\bx^2) < \gamma(\bx^1)$ and $\lambda(\bx^2)
< \lambda(\bx^1)$, we have from (\ref{eq:appen5a-4}) the desired
inequality 
$\partial e(\bx^1) / \partial y_{d^2} < \partial e(\bx^2) / \partial y_{d^2}$.

\section{Proof of Theorem~\ref{thm:8}}	\label{appen:thm8}

Let $\bx^\star = \bN(\bm, K, \beta_{IA})$ be an NE and $\by^\star 
= \by^\star(\bm, K, \beta_{IA})$ for notational convenience. Also, 
we write $e(\by^\star)$ in place of $e(\bX(\by^\star))$. 
By slightly rewriting the social costs given by  (\ref{eq:SocialCost1}) and 
(\ref{eq:SocialCost2}), we obtain
\beqan
\tC(\by
^\star)
\myeq (c_P - \Delta L \ \tau_{DA}) \sum_{d \in \cD} y_d^{\star} 
	+ \tau_{DA} \ L_U  \lb
&& \myhb + \tau_{DA} \ L_U \ d_{\avg} \ e(\by^\star)  
 	- \Delta L \ \tau_{DA} \ e(\by^\star) \sum_{d \in \cD} d \cdot y_d^{\star}
\eeqan
and 
\beqan
SC( \bx^\star ) 
\myeq \big( c_P - \Delta L \ \tau_{DA} \big) 
	\sum_{d \in \cD} x_{d,P}^\star + \tau_{DA} \ L_U  \lb
&& \myhb + \tau_{DA} \ L_U \ d_{\avg} \ e(\bx^\star) 
 	- \Delta L \ \tau_{DA} \ e(\bx^\star) \sum_{d \in \cD} d \cdot 
	x_{d, P}^\star. 
\eeqan
We first derive an upper bound on the difference $SC(\bx^\star) - \tC(\by^\star)$
followed by a lower bound on $\tC(\by^\star)$. 

Subtracting $\tC(\by^\star)$ from $SC(\bx^\star)$, 
\beqan
&& \myhb SC(\bx^\star) - \tC(\by^\star) \lb
\myeq \big( c_P - \Delta L \ \tau_{DA} \big) \Big(  \sum_{d \in \cD} x_{d,P}^\star 
	- \sum_{d \in \cD} y_d^{\star} \Big) \lb
&& \myb + \Delta L \ \tau_{DA} \Big( e(\by^\star) \sum_{d \in \cD} ( y_d^{\star} \cdot d) 
	- e(\bx^\star) \sum_{d \in \cD} ( x_{d,P}^\star \cdot d)   \Big) \lb
&& \myb + \tau_{DA} \ L_U \ d_{\avg} \big( e(\bx^\star) - e(\by^\star)  \big).
\eeqan
From Theorems \ref{thm:1}, \ref{thm:5} and \ref{thm:7}, we know 
$x^\star_{d,P} \leq y^\star_d$. Hence, together with the assumption 
$c_P \geq \Delta L \ \tau_{DA}$, we get
\beqa
&& \myhb SC(\bx^\star) - \tC(\by^\star) \lb
\myleq \Delta L \ \tau_{DA} \Big( e(\by^\star)  \sum_{d \in \cD} (y_d^{\star} \cdot d)
	- e(\bx^\star) \sum_{d \in \cD} (x_{d,P}^\star \cdot d) \Big) \lb
&& + \tau_{DA} \ L_U \ d_{\avg} \big( e(\bx^\star) - e(\by^\star) \big). 
	\label{eq:appen8-0}
\eeqa

We consider the following two cases.
\\ \vspace{-0.1in}

{\em Case 1: $e(\by^\star)  \sum_{d \in \cD} (y_d^{\star} \cdot d)
	\geq e(\bx^\star) \sum_{d \in \cD} (x_{d,P}^\star \cdot d)$} -- In this case, 
we have 
\beqa
(\ref{eq:appen8-0})
\myleq L_U \  \tau_{DA} \Big( e(\by^\star)  \sum_{d \in \cD} (y_d^{\star} \cdot d)
	- e(\bx^\star) \sum_{d \in \cD} (x_{d,P}^\star \cdot d) \Big) \lb
&& + \tau_{DA} \ L_U \ d_{\avg} \big( e(\bx^\star) - e(\by^\star) \big) \lb
\myeq L_U \  \tau_{DA} \Big[ e(\bx^\star)  \Big( d_{\avg} - 
	\sum_{d \in \cD} (x_{d, P}^\star \cdot d) \Big)  \lb
&& \hspace{0.4in} 
	- e(\by^\star)  \Big( d_{\avg} - \sum_{d \in \cD} (y_{d}^\star \cdot d) \Big) 
		\Big].
	\label{eq:appen8-3}
\eeqa
From (\ref{eq:appen8-3}) it is clear that the maximum is achieved when
$\bx^\star = {\bf 0}$ and $\by^\star = \bm$. Hence, 
\beqan
(\ref{eq:appen8-0})
\myleq  \tau_{DA} \ L_U \ d_{\avg} \ e_{\max}(\bm, K, \beta_{IA}).
\eeqan

{\em Case 2: $e(\by^\star)  \sum_{d \in \cD} (y_d^{\star} \cdot d)
	< e(\bx^\star) \sum_{d \in \cD} (x_{d,P}^\star \cdot d)$} -- Under the 
assumption, it is obvious
\beqan
(\ref{eq:appen8-0})
\myleq \tau_{DA} \ L_U \ d_{\avg} \big( e(\bx^\star) - e(\by^\star) \big) \lb
\myleq \tau_{DA} \ L_U \ d_{\avg} \ e_{\max}(\bm, K, \beta_{IA}). 
\eeqan

From these two cases, it is clear that $\tau_{DA} \ L_U \ d_{\avg} \
e_{\max}(\bm, K, \beta_{IA})$ 
is an upper bound for $SC(\bx^\star) - \tC(\by^\star)$. 

Since we assume $c_P \geq \Delta L \ \tau_{DA}$,
we get the following lower bound on $\tC(\by^\star)$. 
\beqan
\tC(\by^\star)
\mygeq \tau_{DA} \ L_U 
+  \tau_{DA} \ e(\by^\star) \big( L_U \cdot d_{\avg} 
	- \Delta L \sum_{d \in \cD} d \cdot y_d^{\star} \big) \lb
\mygeq \tau_{DA} \ L_U, 
\eeqan
where the second inequality is a consequence of $L_U \cdot d_{\avg} 
\geq \Delta L \sum_{d \in \cD} d \cdot y_d^{\star}$. 

Using the above upper bound on $SC(\bx^\star) - \tC(\by^\star)$ and the lower 
bound on $\tC(\by^\star)$, 
\beqan
\frac{SC(\bx^\star)}{\tC(\by^\star)}
\myeq 1 + \frac{SC(\bx^\star) - \tC(\by^\star)}{\tC(\by^\star)} \lb
\myleq 1 + \frac{ \tau_{DA} \ L_U \ d_{\avg} \ 
	e_{\max}(\bm, K, \beta_{IA})}{  \tau_{DA} \ L_U } 
	\lb
\myeq 1 + d_{\avg} \ e_{\max}(\bm, K, \beta_{IA}). 
\eeqan

\end{document}